\documentclass[12pt]{article}

\usepackage{graphicx,eepic,epic,fullpage}
\usepackage[colorlinks]{hyperref}
\usepackage[header,title]{appendix}
\usepackage{latexsym}
\usepackage{amssymb,amsthm}
\usepackage[reqno,leqno]{amsmath} 
\usepackage{bbm}
\usepackage{euscript}
\usepackage{subfigure}
\usepackage{color}
\usepackage{midfloat}
\usepackage{wrapfig}
\usepackage{mdwlist}
\newtheorem{definition}{Definition}
\newtheorem{lemma}{Lemma}
\newtheorem{proposition}{Proposition}
\newtheorem{theorem}{Theorem}
\newtheorem{corollary}{Corollary}
\newtheorem*{lema}{Lemma~\ref{lemma.main}}
\newtheorem*{prop1}{Proposition~\ref{prop.random_sparsity}}
\newcommand{\cA}{\mathcal{A}}
\newcommand{\R}{\mathbb{R}}
\newcommand{\ind}{\mathbf{1}}

\newcommand{\vertiii}[1]{{\left\vert\kern-0.25ex\left\vert\kern-0.25ex\left\vert #1 
    \right\vert\kern-0.25ex\right\vert\kern-0.25ex\right\vert}}
    \makeatother

\title{Sketching Sparse Matrices}
\author{Gautam Dasarathy, Parikshit Shah, Badri Narayan Bhaskar, and Rob Nowak\\ University of Wisconsin - Madison}
\date{March 26, 2013}

\begin{document}
\maketitle

\begin{abstract}
{This paper considers the problem of recovering an unknown sparse $p\times p$ matrix $X$ from an $m\times m$ matrix $Y=AXB^T$, where
$A$ and $B$ are known $m \times p$ matrices with $m\ll p$.  

The main result shows that there exist constructions of the ``sketching'' matrices $A$ and $B$ so that even if $X$ has $\mathcal{O}(p)$ non-zeros, it can be recovered exactly and efficiently using a convex program as long as these non-zeros are not concentrated in any single row/column of $X$. Furthermore, it suffices for the size of $Y$ (the sketch dimension) to scale as $m = \mathcal{O}\left(\sqrt{\mbox{ \# nonzeros in }X} \times\log p\right)$. The results also show that the recovery is robust and stable in the sense that if $X$ is equal to a sparse matrix plus a perturbation, then the convex program we propose produces an approximation with accuracy proportional to the size of the perturbation. Unlike traditional results on sparse recovery, where the sensing matrix produces independent measurements, our sensing operator is highly constrained (it assumes a tensor product structure). Therefore, proving recovery guarantees require non-standard techniques. Indeed our approach relies on a novel result concerning tensor products of bipartite graphs, which may be of independent interest. 

This problem is motivated by the following application, among others.  Consider a $p \times n$ data matrix $D$, consisting of $n$ observations of $p$  variables.  Assume that the correlation matrix $X:=DD^{T}$ is (approximately) sparse in the sense that each of the $p$ variables is significantly correlated with only a few others.  Our results show that these significant correlations can be detected even if we have access to only a sketch of the data  $S=AD$ with $A \in R^{m\times p}$.}
\end{abstract}
\section{Introduction}
\label{sec.introduction}
An important feature of many modern data analysis problems is the presence of  a large number of variables relative to the amount of available resources. Such high dimensionality occurs in a range of applications in bioinformatics, climate studies, and economics. Accordingly, a fruitful and active research agenda over the last few years has been the development of methods for sampling, estimation, and learning that take into account \emph{structure} in the underlying model and thereby making these problems tractable. A notion of structure that has seen many applications is that of \emph{sparsity}, and methods for sampling and estimating sparse signals have been the subject of intense research in the past few years \cite{candes2005decoding,candes2006robust,donoho2006compressed}

In this paper we will study a more nuanced notion of structure which we call  \emph{distributed sparsity}. For what follows, it will be convenient to think of the unknown high-dimensional signal of interest as being represented as a matrix $X$. Roughly, the signal is said to be distributed sparse if every row and every column of $X$ has only a few non-zeros. We will see that it is possible to design efficient and effective acquisition and estimation mechanisms for such signals. Let us begin by considering a few example scenarios where one might encounter distributed sparsity.

\begin{itemize}
\item \textbf{Covariance Matrices:} Covariance matrices associated to some natural phenomena have the property that each covariate is correlated with only a few other covariates. For instance, it is observed that protein signaling networks are such that there are only a few significant correlations \cite{sachs2005causal} and hence the discovery of the such networks from experimental data naturally leads to the estimation of a covariance matrix (where the covariates are proteins) which is (approximately) distributed sparse. Similarly, the covariance structure corresponding to longitudinal data is distributed sparse \cite{diggle1998nonparametric}. See Section~\ref{sec.covarianceSketching}.
\item \textbf{Multi-dimensional signals:} Multi-dimensional signals such as the natural images that arise in medical imaging \cite{candes2005decoding} are known to be sparse in the gradient domain. When the features in the images are not axis-aligned, not only is the matrix representation of the image gradient sparse, it is also distributed sparse. For a little more on this, see Section~\ref{sec.multiDimensionalSignalProcessing}
\item \textbf{Random Sparse Signals and Random Graphs:}
Signals where the sparsity pattern is \emph{random}  (i.e., each entry is nonzero independently and with a probability $q$) are also distributed sparse with high probability. The ``distributedness'' of the  sparsity pattern can be measured using the ``degree of sparsity'' $d$ which is defined to be the maximum number of non-zeros in any row or column. For random sparsity patterns, we have the following:
\begin{proposition} \label{prop.random_sparsity}
Consider a random matrix $X \in\mathbb{R}^{p \times p}$ whose entries are independent copies of the Bernoulli$(\gamma)$ \footnote{Recall that if $\chi\sim\mbox{ Bernoulli}(\gamma)$, then $P(\chi=1) = \gamma$ and $P(\chi =0) = 1-\gamma$.}distribution where $p\gamma = \Delta= \Theta(1)$. Then for any $\epsilon>0$, $X$ has at most $d$ 1's in each row/column with probability at least $1-\epsilon$, where $$d = \Delta\left(1 + {\frac{2\log(2p/\epsilon)}{\Delta}}\right).$$
\end{proposition}
(The proof is straightforward, and available in Appendix \ref{appendix.B}.)

In a similar vein, combinatorial graphs have small \emph{degree} in a variety of applications, and their corresponding matrix representation will then be distributed sparse. For instance, Erdos-Renyi random graphs $\mathcal{G}(p,q)$ with $pq=\mathcal{O}(\log p)$ have small degree \cite{bollobas2001random}.
\end{itemize}


\subsection{Problem Setup and Main Results}
\label{subsec.setup}

Our goal is to invert an underdetermined linear system of the form \begin{equation}
\label{eq.linearSystem}
Y = AX B^T , 
\end{equation}
where $A= [a_{ij}]\in\mathbb{R}^{m\times p}, B=[b_{ij}]\in\mathbb{R}^{m\times p}$, with $m\ll p$ and $X\in\mathbb{R}^{p\times p}$. Since the matrix $X\in\mathbb{R}^{p\times p}$ is linearly transformed to obtain the smaller dimensional matrix $Y\in\mathbb{R}^{m\times m}$, we will refer to $Y$ as the \emph{sketch} (borrowing terminology from the computer science literature \cite{muthukrishnan2005data}) of $X$ and we will refer to the quantity $m$ as the \emph{sketching dimension}. Since the value of $m$ signifies the amount of compression achieved, it is desirable to have as small a value of $m$ as possible. 

Rewriting the above using tensor product notation, with $y = $ vec$(Y)$ and $x =$ vec$(X)$, we equivalently have
\begin{equation}
y = (B\otimes A)x, 
\end{equation}
where vec$(X)$ simply \emph{vectorizes} the matrix $X$, i.e., produces a long column vector by stacking the columns of the matrix and $B\otimes A$ is the tensor (or Kronecker) product of $B$ and $A$, given by
\begin{equation}
\left[\begin{array}{cccc}
b_{11}A & b_{12}A&\cdots &b_{1p}A\\
b_{21}A& b_{22}A&\cdots &b_{2p}A\\
\vdots&\vdots&\ddots&\vdots\\
b_{m1}A& b_{m2}A&\cdots &b_{mp}A\\
\end{array}
\right].
\end{equation}

While it is not possible to invert such underdetermined systems of equations in general, the rapidly growing literature on what has come to be known as \emph{compressed sensing} suggests that this can be done under certain assumptions. In particular, taking cues from this literature, one might think that this is possible if $x$ (or equivalently $X$) has only a few non-zeros. 

Let us first consider the case when there are only $k = \Theta(1)$ non-zeros in $X$, i.e., it is very sparse. Then, it is possible to prove that the optimization program~\eqref{eq.optP1} recovers $X$ from $AX B^T$ using  standard ``RIP-based'' techniques \cite{candes2005decoding}. We refer the interested reader to the papers by Jokar et al \cite{jokar2009sparse} and Duarte et al \cite{duarte2012kronecker} for more details, but in essence the authors show that if $\delta_r(A)$ and $\delta_r(B)$ are the restricted isometry constants (of order $r$) \cite{candes2005decoding} for $A$ and $B$ respectively, then the following is true about $B\otimes A$
\begin{equation*}
\max\left\{\delta_r(A),\delta_r(B)\right\}\leq \delta_r(A\otimes B) = \delta_r(B\otimes A) \leq \left(1+\delta_r(A)\right)\left(1+\delta_r(B)\right) - 1.
\end{equation*}

In many interesting problems that arise naturally, as we will see in subsequent sections, a more realistic assumption to make is that $X$ has $\mathcal{O}(p)$ non-zeros and it is this setting we consider for this paper.  Unfortunately, the proof techniques outlined above cannot succeed in such a demanding scenario. As hinted earlier, it will turn out however that one cannot handle arbitrary sparsity patterns and that the non-zero pattern of $X$ needs to be \emph{distributed}, i.e., each row/column of $X$ cannot have more than a few, say $d$, non-zeros. We will call such matrices $d-$distributed sparse (see Definition~\ref{def.distributedSparse}). We explore this notion of structure in more detail in Section~\ref{sec.distributedSparisity}.

An obvious, albeit highly impractical, approach to recover a (distributed) sparse $X$ from measurements of the form $Y = AX B^T$ is the following: search over all matrices $\tilde{X}\in\mathbb{R}^{p\times p}$ such that $A\tilde{X}B^T$ agrees with $Y = AX B^T$ and find the sparsest one. One might hope that under reasonable assumptions, such a procedure would return $X$ as the solution. However, there is no guarantee that this approach might work and worse still, such a search procedure is known to be computationally infeasible.

We instead consider solving the optimization program \eqref{eq.optP1} which is a natural (convex) relaxation of the above approach.
\begin{equation}
\begingroup
\begin{aligned}
\underset{\tilde{X}}{\mbox{minimize}}&\qquad\left\|\tilde{X}\right\|_1\\
\mbox{subject to} &\qquad A\tilde{X}B^T = Y.
\end{aligned}\tag{P$_1$}\label{eq.optP1}
\endgroup
\end{equation}
Here, by $\|\tilde{X}\|_1$ we mean $\sum_{i,j}\left|\tilde{X}_{i,j}\right|$, i.e., the $\ell_1$ norm of vec$(\tilde{X})$. 

The {main part} of the paper is devoted to showing with high probability \eqref{eq.optP1} has a unique solution that equals $X$. In particular, we prove the following result. 

\begin{theorem}
\label{thm.main}
Suppose that $X$ is $d-$distributed sparse. Also, suppose that $A,B\in\left\{0,1\right\}^{m\times p}$ are drawn independently and uniformly from the $\delta-$random bipartite ensemble\footnote{Roughly speaking, the $\delta-$random bipartite ensemble consists of the set of all 0-1 matrices that have almost exactly $\delta$ ones per column. We refer the reader to Definition~\ref{def.randomGraph} for the precise definition and Section~\ref{sec.mainLemma} for more details.}. Then as long as $${m=\mathcal{O}(\sqrt{dp}\log p)}\;\;\;\;\mbox{and}\;\;\;\; \delta = \mathcal{O}(\log p),\;\;\;$$ there exists a $c>0$ such that the optimal solution $X^\ast$ of \eqref{eq.optP1} equals $X$ with probability exceeding $1-p^{-c}$. Furthermore, this holds even if $B$ equals $A$.
\end{theorem}

In Section~\ref{sec.proof}, we will prove Theorem~\ref{thm.main} for the case when $B = A$. It is quite straightforward to modify this proof to the case where $A$ and $B$ are independently drawn, since there is much more independence that can be leveraged. 

 Let us pause here and consider some implications of this theorem.
\begin{enumerate}
\item \eqref{eq.optP1} does not impose any structural restrictions on $X^\ast$. In other words, even though $X$ is assumed to be distributed sparse, this (highly non-convex) constraint need not be factored in to the optimization problem.  This ensures that $\eqref{eq.optP1}$ is a Linear Program (see e.g., \cite{bertsimas1997introduction}) and can thus be solved efficiently. 

\item Recall that what we observe can be thought of as the $\mathbb{R}^{m^2}$ vector \\$(B\otimes A)x$. Since $X$ is $d-$distributed sparse, $x$ has $\mathcal{O}(dp)$ non-zeros. Now, even if an oracle were to reveal the exact locations of these non-zeros, we would require at least $\mathcal{O}(dp)$ measurements to be able to perform the necessary inversion to recover $x$. In other words, it is absolutely necessary for $m^2$ to be at least $\mathcal{O}(dp)$. Comparing this to Theorem~\ref{thm.main} shows that the simple algorithm we propose is \emph{near optimal} in the sense that it is only a logarithm away from this trivial lower bound. This logarithmic factor also makes an appearance in the measurement bounds in the compressed sensing literature \cite{candes2006robust}. 

\item Finally, as mentioned earlier, inversion of under-determined linear systems where the linear operator assumes a tensor product structure has been studied earlier \cite{jokarKroneckerCS2010, duarteKroneckerCS2010}. However, these methods are relevant only in the regime where the sparsity of the signal to be recovered is much smaller than the dimension $p$. The proof techniques they employ will unfortunately not allow one to handle the more demanding situation the sparsity scales linearly in $p$ and if one attempted an extension of their techniques naively to this situation, one would see that the sketch size $m$ needs to scale like $\mathcal{O}(dp\log^2 p)$ in order to recover a $d-$ distributed sparse matrix $X$. This is of course uninteresting since it would imply that the size of the sketch is bigger than the size of $X$.
\end{enumerate}

It is possible that $Y$ was not exactly observed, but rather is only available to us as a corrupted version $\hat{Y}$. For instance, $\hat{Y}$ could be $Y$ corrupted by independent zero mean, additive Gaussian noise or in case of the covariance sketching problem discussed in Section~\ref{sec.covarianceSketching}, $\hat{Y}$ could be an empirical estimate of the covariance matrix $ Y = AX A^T$. In both these cases, a natural relaxation to \eqref{eq.optP1} would be the following optimization program~\eqref{eq.optP2} (with $B$ set to $A$ in the latter case). 
\begin{equation}
\begingroup
\begin{aligned}
\underset{\tilde{X}}{\mbox{minimize}}&\qquad \| A\tilde{X}B^T - \hat{Y} \|_2^2 + \lambda \left\|\tilde{X}\right\|_1
\end{aligned}\tag{P$_2$}\label{eq.optP2}
\endgroup
\end{equation}
Notice that if $X$ was a sparse covariance matrix and if $A=B=I_{p\times p}$, then~\eqref{eq.optP2} reduces to the soft thresholding estimator of sparse covariance matrices studied in \cite{rothman2009generalized}. 

While our experimental results show that this optimization  program \eqref{eq.optP2} performs well, we leave its exploration and analysis to future work. We will instead state the following ``approximation'' result that shows that the solution of \eqref{eq.optP1} is close to the optimal $d-$ distributed sparse approximation for any matrix $X$. The proof is similar to the proof of Theorem 3 in \cite{berindeIndyk2008} and is provided in Appendix~\ref{appendix:C}. Given $p\in\mathbb{N}$, let $[p]$ denote the set $\left\{1,2,\ldots,p\right\}$ and  let $\mathfrak{W}_{d,p}$ denote the following collection of subsets of $[p]\times  [p]$:
\begin{align*}
\mathfrak{W}_{d,p}:=&\left\{\Omega \subset [p]\times [p]: \left|\Omega\cap \left\{\left\{i\right\}\times [p]\right\}\right|\leq d,\left|\Omega\cap \left\{[p]\times\{i\}\right\}\right|\leq d,\mbox{ for all }i\in [p]\right\}.
\end{align*}
Notice that if a matrix $X\in\mathbb{R}^{p\times p}$ is such that there exists an $\Omega\in\mathfrak{W}_{d,p}$ with the property that $X_{ij}\neq 0$ only if $(i,j)\in\Omega$, then the matrix is $d-$distributed sparse. 

Given $\Omega\in\mathfrak{W}_{d,p}$ and a matrix $X\in\mathbb{R}^{p\times p}$, we write $X_\Omega$ to denote the projection of $X$ onto the set of all matrices supported on $\Omega$. That is, 
\begin{align*}
\left[X_\Omega\right]_{i,j} = \begin{cases}
X_{i,j} & \mbox{if } (i,j)\in\Omega \\
0 & \mbox{otherwise}
\end{cases}\;\;\;\mbox{for all }(i,j)\in[p]\times [p].
\end{align*}
\begin{theorem}
\label{thm.approximation}
Suppose that $X$ is an arbitrary $p\times p$ matrix and that the hypotheses of Theorem~\ref{thm.main} hold. Let $X^\ast$ denote the solution to the optimization program~\eqref{eq.optP1}. Then, there exist constants $c>0$ and $\epsilon\in(0,1/4)$ such that the following holds with probability exceeding $1-p^{-c}$. 
\begin{equation}
\left\|X^\ast - X\right\|_1 \leq \frac{2 - 4\epsilon}{1-4\epsilon}\left(\min_{\Omega\in\mathfrak{W}_{d,p}} \left\|X - X_\Omega\right\|_1\right).
\end{equation}
\end{theorem}

The above theorem tells us that even if $X$ is not structured in any way, the solution of the optimization program \eqref{eq.optP1} approximates $X$ as well as the best possible $d-$distributed sparse approximation of $X$ (up to a constant factor). This has interesting implications, for instance, to situations where a $d-$distributed sparse $X$ is corrupted by a ``noise'' matrix $N$ as shown in the following corollary. 
\begin{corollary}
\label{cor.noisySigma}
Suppose $X\in\mathbb{R}^{p\times p}$ is $d-$distributed sparse and suppose that $\hat{X} = X + N$. Then, the solution $X^\ast$ to the optimization program $$\underset{\tilde{X}}{\mbox{min. }}\left\|\tilde{X}\right\|_1\mbox{ subject to }A\tilde{X}B^T =A\hat{X}B^T$$
satisfies
\begin{equation}
\left\|X^\ast - X\right\|_1\leq \frac{5-12\epsilon}{1-4\epsilon}\left\|N\right\|_1
\end{equation}
\end{corollary}
\begin{proof}
Let $\Omega$ be the support of $X$. 
To prove the result, we will consider the following chain of inequalities. 
\begin{align}
\left\|X^\ast - X\right\|_1 &\leq \left\|X^\ast - \hat{X}\right\|_1 + \left\| \hat{X}-X\right\|_1\notag\\
&\stackrel{(a)}{\leq} \frac{2-4\epsilon}{1-4\epsilon}\left\|\hat{X}-\hat{X}_\Omega\right\|_1 + \left\|\hat{X}-X\right\|_1\notag\\
&\leq \frac{2-4\epsilon}{1-4\epsilon}\left\|\hat{X}-X\right\|_1 + \frac{2-4\epsilon}{1-4\epsilon}\left\|\hat{X}_\Omega-X\right\|_1 + \left\|\hat{X}-X\right\|_1\notag\\
&\stackrel{(b)}{\leq} \frac{5-12\epsilon}{1-4\epsilon}\left\|\hat{X}-X\right\|_1\notag\\
&\stackrel{(c)}{=} \frac{5-12\epsilon}{1-4\epsilon}\left\|N\right\|_1.\notag
\end{align}
Here $(a)$ follows from Theorem~\ref{thm.approximation} since $\Omega\in\mathfrak{W}_{d,p}$ and $(b)$ follows from the fact that \\$\left\|\hat{X}_\Omega-X\right\|\leq\left\|\hat{X}-X\right\|$ since $X_{\Omega^c}$ is $\mathbf{0}_{p\times p}$. Finally, in $(c)$ we merely plug in the definition of $\hat{X}$. 
\end{proof}

\subsection{The Rectangular Case and Higher Dimensional Signals}
\label{sec.rectangular}
While Theorem~\ref{thm.main}, as stated, applies only to the case of square matrices $X$, we can extend our result in a straightforward manner to the rectangular case. Consider a matrix $X \in \R^{p_1 \times p_2}$ where (without loss of generality) $p_1 < p_2$. We assume that the row degree is $d_r$ (i.e. no row of $X$ has more than $d_r$ non-zeros) and that the column degree is $d_c$. Consider sketching matrices $A \in \R^{m \times p_1}$ and $B \in \R^{m \times p_2}$ and the sketching operation:
$$
Y=AXB^T.
$$
Then we have the following corollary:
\begin{corollary}
\label{cor.rectangular}
Suppose that $X$ is distributed sparse with row degree $d_r$ and colum degree $d_c$. Also, suppose that $A\in\left\{0,1\right\}^{m\times p_1}$, $B\in\left\{0,1\right\}^{m\times p_2}$ are drawn independently and uniformly from the $\delta-$random bipartite ensemble. Let us define $p=\max (p_1, p_2)$ and $d=\max (d_r, d_c)$.

Then if $${m=\mathcal{O}(\sqrt{dp}\log p)}\;\;\;\;\mbox{and}\;\;\;\; \delta = \mathcal{O}(\log p),\;\;\;$$ there exists a $c>0$ such that the optimal solution $X^\ast$ of \eqref{eq.optP1} equals $X$ with probability exceeding $1-p^{-c}$. 
\end{corollary}
\begin{proof}
Let us define the matrix $\tilde{X} \in \R^{p \times p}$ as 
$$
\tilde{X} = \left[ \begin{array}{c} X \\ 0 \end{array} \right],
$$
i.e. it is made square by padding additional zero rows.
Note that $\tilde{X}$ has degree $d=\max (d_r, d_c)$.
Moreover note that the matrix $A \in \R^{m \times p_1}$ can be augmented to $\tilde{A} \in \R^{m \times p}$ via:
$$
\tilde{A} = \left[ \begin{array}{cc} 
A & \bar{A}
\end{array} \right]
$$
where $\bar{A} \in \R^{m \times (p-p_1)}$ is also drawn from the $\delta$-random bipartite ensemble. Then one has the relation:
$$
Y=\tilde{A}\tilde{X}B^T.
$$ 
Thus, the rectangular problem can be reduced to the standard square case considered in Theorem~\ref{thm.main}, and the result follows.

\end{proof}

The above result shows that a finer analysis is required for the rectangular case. For instance, if one were to consider a scenario where $p_1 = 1$, then from the compressed sensing literature, we know that the result of Corollary~\ref{cor.rectangular} is weak. We believe that determining the right scaling of the sketch dimension(s) in the case when $X$ is rectangular is an interesting avenue for future work. 

Finally, we must also state that while the results in this paper only deal with two-dimensional signals, similar techniques can be used to deal with higher dimensional tensors that are distributed sparse. We leave a detailed exploration of this question to future work.


\subsection{Applications}
\label{sec.applications}
It is instructive at this stage to consider a few examples of the framework we set up in this paper. These applications demonstrate that the modeling assumptions we make viz., tensor product sensing and distributed sparsity are important and arise naturally in a wide variety of contexts.
\subsubsection{Covariance Estimation from Compressed realizations \\or Covariance Sketching}
\label{sec.covarianceSketching}
One particular application that will be of interest to us is the estimation of covariance matrices from sketches of the sample vectors. We call this \emph{covariance sketching}.

Consider a scenario in which the covariance matrix $\Sigma\in\mathbb{R}^{p\times p} $ of a high-dimensional zero-mean random vector $\xi = (\xi_1, \ldots, \xi_p) ^T$ is to be estimated. In many applications of interest, one determines $\Sigma$ by conducting correlation tests for each pair of covariates $\xi_i, \xi_j$ and computing an estimate of $\mathbf{E}[\xi_i \xi_j]$ for $i, j = 1, \ldots, p$. This requires one to perform correlation tests for $\mathcal{O}(p^2)$ pairs of covariates, a daunting task in the high-dimensional setting. Perhaps most importantly, in many cases of interest, the underlying covariance matrix may have structure, which such an approach may fail to exploit. For instance if $\Sigma$ is very sparse, it would be vastly more efficient to perform correlation tests corresponding to only the non-zero entries. The chief difficulty of course is that the sparsity pattern is rarely known in advance, and finding this is often the objective of the experiment.

In other settings of interest, one may obtain statistical samples by observing $n$ independent \emph{sample paths} of the statistical process. When $\xi$ is high-dimensional, it may be infeasible or undesirable to sample and store the entire sample paths $\xi^{(1)}, \ldots, \xi^{(n)} \in \mathbb{R}^p$, and it may be desirable to reduce the dimensionality of the acquired samples.

Thus in the high-dimensional setting we propose an alternative acquisition mechanism: pool covariates together  to form a collection of new variables $Z_1, \ldots, Z_m$, where $m<p$. For example one may construct:
$$
Z_1=\xi_1+\xi_2+\xi_6, \;\; Z_2=\xi_1+\xi_4+\xi_8+\xi_{12}, \;\;\ldots
$$
and so on; more generally we have measurements of the form $Z=A\xi$ where $A \in \R^{m \times p}$ and typically $m\ll p$. We call the thus newly constructed covariates $Z = \left(Z_1, \ldots, Z_m\right)$ a sketch of the random vector $\xi$.

More formally, the covariance sketching problem can be stated as follows. 
Let $\xi^{(1)},\xi^{(2)},\ldots,\xi^{(n)}\in\mathbb{R}^{p}$ be $n$ independent and identically distributed $p-$variate random vectors and let $\Sigma\in\mathbb{R}^{p\times p}$ be their unknown covariance matrix. Now, suppose that one has access to the $m-$dimensional \emph{sketch vectors} $Z^{(i)}$ such that $$Z^{(i)} = A\xi^{(i)}, \; \; i =1,2,\ldots,n,$$
where $A\in\mathbb{R}^{m\times p}, m<p$ is what we call a \emph{sketching matrix}. The goal then is to recover $\Sigma$ using only $\{Z^{(i)}\}_{i=1}^n$. The sketching matrices we will focus on later will have randomly-generated binary values, so each element of $Z^{(i)}$ will turn out to be a sum (or ``pool'') of a random subset of the covariates.

Notice that the sample covariance matrix computed using the vectors $\left\{Z^{(i)}\right\}_{i=1}^n$ satisfies the following. 
\begin{align*}
\hat{\Sigma}^{(n)}_Z 
&:=\frac{1}{n}\sum_{i=1}^nZ^{(i)}(Z^{(i)})^{T}\\
&= A\left(\frac{1}{n}\sum_{i=1}^n\xi^{(i)}(\xi^{(i)})^{T}\right)A^T\\
&= A\hat{\Sigma}^{(n)}A^T,
\end{align*}
where $\hat{\Sigma}^{(n)}:=\frac{1}{n}\sum_{i=1}^n\xi^{(i)}(\xi^{(i)})^{T}$ is the (maximum likelihood) estimate of $\Sigma$ from the samples $\xi^{(1)}\ldots,\xi^{(n)}$. 

To gain a better understanding of the covariance sketching problem, it is natural to first consider the stylized version of the problem suggested by the above calculation. That is, whether it is possible to efficiently recover a matrix $\Sigma\in\mathbb{R}^{p\times p}$ given the ideal covariance matrix of the sketches $\Sigma_Z = A\Sigma A^T\in\mathbb{R}^{m\times m}$. The analysis in the current paper focuses on exactly this problem and thus helps in exposing the most unique and challenging aspects of covariance sketching.

The theory developed in this paper tells us that at the very least, one needs to restrict the underlying random vector $\xi$ to have the property that each $\xi_i$ depends on only a few (say, $d$) of the other $\xi_j$'s. Notice that this would of course imply that the true covariance matrix $\Sigma$ will be $d-$distributed sparse. Applying Theorem~\ref{thm.main}, especially the version in which the matrices $A$ and $B$ are identical, to this stylized situation reveals the following result. If $A$ is chosen from a particular random ensemble and if one gets to observe the covariance matrix $A\Sigma A^T$ of the sketch random vector $Z = A\xi$, then using a very efficient convex optimization program, one can recover $\Sigma$ exactly. 

Now, suppose that $\xi$ and $A$ are as above and that we get to observe $n$ samples $Z^{(i)} = A\xi^{(i)}, i = 1,2,\ldots,n$ of the sketch $Z = A\xi$. 
Notice that we can can consider $\hat{\Sigma}^{(n)}:=\frac{1}{n}\sum_{i=1}^n\xi^{(i)}(\xi^{(i)})^{T}$ to be a ``noise corrupted version'' of $\Sigma$ since  we can write $$\hat{\Sigma}^{(n)} = \Sigma + (\hat{\Sigma}^{(n)} - \Sigma),$$
where, under reasonable assumptions on the underlying distribution, \\$\left\|\hat{\Sigma}^{(n)}-\Sigma\right\|_1\to0$ almost surely as $n\to\infty$ by the strong law of large of numbers. Therefore, an application of Theorem~\ref{thm.approximation} tells us that solving \eqref{eq.optP1} with the observation matrix $\hat{\Sigma}^{(n)}_Z$ gives us an asymptotically consistent procedure to estimate the covariance matrix $\Sigma$ from sketched realizations. 

We anticipate that our results will be interesting in many areas such as 
quantitative biology where it may be possible to naturally pool together covariates and measure interactions at this pool level. Our work shows  that covariance structures that occur naturally are amenable to covariance sketching, so that drastic savings are possible  when correlation tests are performed at the pool level, rather than using individual covariates.

The framework we develop in this paper can also be used to accomplish \emph{cross covariance sketching}. That is, suppose that $\xi$ and $\zeta$ are two zero mean $p-$variate random vectors and suppose that $\Sigma_{\xi\zeta}\in\mathbb{R}^{p\times p}$ is an unknown matrix such that $\left[\Sigma_{\xi\zeta}\right]_{ij} = \mathbb{E}[\xi_i\zeta_j]$. Let $\{\xi^{(i)}\}_{i=1}^n$ and $\{\zeta^{(i)}\}_{i=1}^n$ be $2n$ independent and identically distributed random realizations of $\xi$ and $\zeta$ respectively. The goal then, is to estimate $\Sigma_{\xi\zeta}$ from the $m$ dimensional \emph{sketch vectors} $Z^{(i)}$ and $W^{(i)}$ such that $$Z^{(i)} = A\xi^{(i)}, W^{(i)} = B\zeta^{(i)} \; \; i =1,2,\ldots,n,$$
where $A,B\in\mathbb{R}^{m\times p}, m<p$. 

As above, in the idealized case, Theorem~\ref{thm.main} shows that the cross-covariance matrix $\Sigma_{\xi\zeta}$ of $\xi$ and $\zeta$ can be exactly recovered from the cross-covariance matrix $\Sigma_{ZW} = A\Sigma_{\xi\zeta} B^T$ of the sketched random vectors $W$ and $Z$ as long as $\Sigma_{\xi\zeta}$ is distributed sparse. In the case we have $n$ samples each of the sketched random vectors, an application of Theorem~\ref{thm.approximation} to this problem tells us that \eqref{eq.optP1} is an efficient and asymptotically consistent procedure to estimate a distributed sparse $\Sigma_{\xi\zeta}$ from compressed realization of $\xi$ and $\zeta$. 

We note that the idea of pooling information in statistics, especially in the context of meta analysis is a classical one \cite{hedges1985statistical}. For instance the classical Cohen's $d$ estimate uses the idea of pooling samples obtained from different distributions to obtain accurate estimates of a common variance. While at a high level the idea of pooling is related, we note that our notion is qualitatively different in that we propose pooling covariates themselves into sketches and obtain samples in this reduced dimensional space.

\subsubsection{Graph Sketching}
Large graphs play an important role in many prominent problems of current interest; two such examples are graphs associated to communication networks (such as the internet) and social networks. Due to their large sizes it is difficult to store, communicate, and analyze these graphs, and it is desirable to compress these graphs so that these tasks are easier. The problem of compressing or sketching graphs has recently gained attention in the literature \cite{graph_sketching, network_compression}.

In this section we propose a new and natural notion of compression of a given graph $G=(V,E)$. The resulting ``compressed'' graph is a weighted graph $\hat{G}=(\hat{V},\hat{E})$, where $\hat{V}$ has a much smaller cardinality than $V$. Typically, $\hat{G}$ will be a complete graph, but the edge weights will encode interesting and valuable information about the original graph. 

Partition the vertex set $V=V_1 \cup V_2 \cup \ldots \cup V_m$; in the compressed graph $\hat{G}$, each partition $V_i$ is represented by a node. (We note that this need not necessarily be a disjoint partition, and we allow for the possibility for $V_i \cap V_j \neq \emptyset$.) For each pair $V_i, V_j \in \hat{V}$, the associated edge weight is the total number of edges crossing from nodes in $V_i$ to the nodes in $V_j$ in the original graph $G$. Note that if an index $k \in V_i \cap V_j$, the self edge $(k,k)$ must be included when counting the total number of edges between $V_i$ and $V_j$. (We point out that the edge $(V_k,V_k) \in \hat{E}$ also carries a non-zero weight; and is precisely equal to the number of edges in $G$ that have both endpoints in $V_k$. See Fig. \textcolor{red}{1} for an illustrative example.)

Define $A_i$ to be the (row) indicator vector for the set $V_i$, i.e. 
$$
A_{ij}=\left\{ \begin{array}{ll} 
1 & \text{ if } j \in V_i \\
0 & \text{ otherwise}
\end{array} \right.
$$
If $X$ denotes the adjacency matrix of $G$, then $Y:=A X A^T$ denotes the matrix representation of $\hat{G}$. The sketch $Y$ has two interesting properties:
\begin{itemize}
\item The encoding faithfully preserves high-level ``cut'' information  about the original graph. For instance information such as the weight of edges crossing between the partitions $V_i$ and $V_j$ is faithfully encoded. This could be useful for networks where the vertex partitions have a natural interpretation such as geographical regions; questions about the total network capacity between two regions is directly available via this method of encoding a graph. Approximate solutions to related questions such as maximum flow between two regions (partitions) can also be provided by solving the problem on the compressed graph. 
\item When the graph is bounded degree, the results in this paper show that there exists a suitable random partitioning scheme such that the proposed method of encoding the graph is lossless. Moreover, the original graph $G$ can be unravelled from the smaller sketched graph $\hat{G}$ efficiently using the convex program \eqref{eq.optP1}.

\end{itemize}
\begin{figure}[h!]
  \centering
    \includegraphics[scale=0.4]{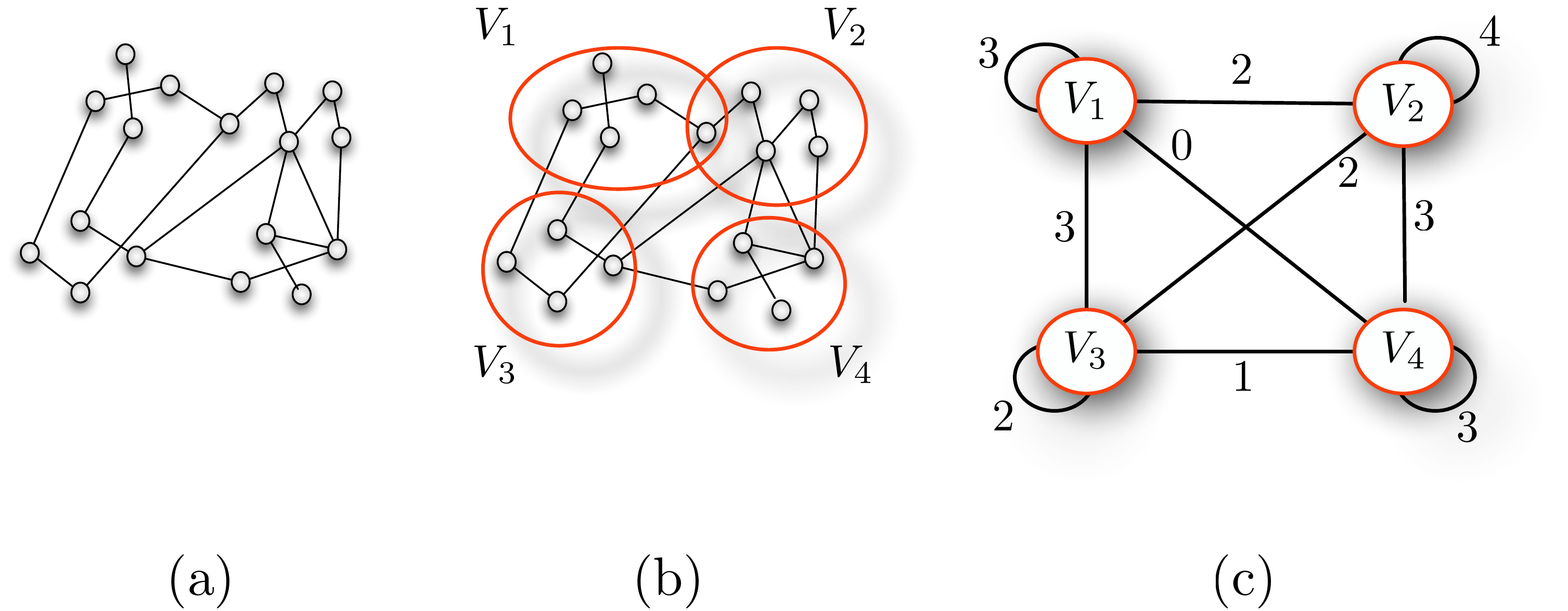}
\caption{An example illustrating graph sketching. (a) A graph $G$ with $17$ nodes (b) Partitioning the nodes into four partitions $V_1, V_2, V_3, V_4$ (c) The sketch of the graph $G$. The nodes represent the partitions and the edges in the sketch represent the total number of edges of $G$ that cross partitions.}
\end{figure}
\subsubsection{Multidimensional Signal Processing}
\label{sec.multiDimensionalSignalProcessing}
Multi-dimensional signals arise in a variety of applications, for instance images are naturally represented as two-dimensional signals $f(\cdot,\cdot)$ over some given domain. 

Often  it is more convenient to view the signal not in the original domain, but rather in a transformed domain. Given some one dimensional family of ``mother" functions $\psi_u(t)$ (usually an orthonormal family of functions indexed with respect to the transform variable $u$), such a family induces the transform for a one dimensional signal $f(t)$ (with domain $\mathcal{T}$) via
$$
\hat{f}(u)=\int_{t \in \mathcal{T}} f(t) \psi_u(t).
$$
For instance if $\psi_u(t):= \exp(-i2 \pi u t)$, this is the Fourier transform, and if $\psi_u(t)$ is chosen to be a wavelet function (where $u=(a,b)$, the translation and scale parameters respectively) this generates the well-known wavelet transform that is now ubiquitous in signal processing.

Using $\psi_u(t)$ to form an orthonormal basis for one-dimensional signals, it is straightforward to extend to a basis for two-dimensional signal by using the functions $\psi_u(t) \psi_v(r)$. Indeed, this defines a two-dimensional transform via
$$
\hat{f}(u,v)=\int_{(t, r) \in \mathcal{X} \times \mathcal{X} } f(t,r) \psi_u(t) \psi_v(r).
$$
Similar to the one-dimensional case, appropriate choices of $\psi$ yeild standard transforms such as the two-dimensional Fourier transform and the two-dimensional wavelet transform. The advantage of working with an alternate basis as described above is that signals often have particularly simple representations when the basis is appropriately chosen. It is well-known, for instance, that natural images have a sparse representation in the wavelet basis (see Fig. \ref{fig:wavelet}). Indeed, many natural images are not only sparse, but they are also distributed sparse, when represented in the wavelet basis. This enables compression by performing ``pooling'' of wavelet coefficients, as described below.

\begin{figure}[h!]
  \centering
    \includegraphics[trim=10mm 40mm 10mm 30mm, clip, scale=0.7]{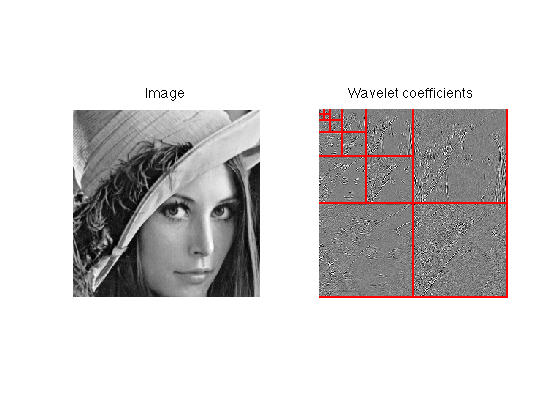}
\caption{Note that the wavelet representaion of the image is distributed sparse.}
\label{fig:wavelet}
\end{figure}

In many applications, it is more convenient to work with discrete signals and their transforms (by discretizing the variables $(t,r)$ and the transform domain variables $(u,v)$). It is natural to represent the discretization of the two-dimensional signal $f(t,r)$ by a matrix $F \in \R^{p \times p}$. The corresponding discretization of $\psi_u(t)$ can be represented as a matrix $\Psi=[\Psi]_{ut}$, and the discretized version of the $\hat{f}(u,v)$, denoted by $\hat{F}$ is given  by:
$$
\hat{F}=\Psi F \Psi^T.
$$

As noted above, in several applications of interest, when the basis $\Psi$ is chosen appropriately, the signal has a succinct representation and the corresponding matrix $\hat{F}$ is sparse. This is true, for instance, when $F$ represents a natural image and $\hat{F}$ is the wavelet transform of $F$. Due to the sparse representability of the signal in this basis, it is possible to acquire and store the signal in a \emph{compressive} manner. For instance, instead of sensing the signal $F$ using $\Psi$ (which correpsonds to sensing the signal at every value of the transform variable $u$), one could instead form ``pools'' of transform variables $S_i = \left\{u_{i1}, u_{i2}\ldots, u_{ik}\right\}$ and sense the signal via
$$
A \Psi = \left[ \begin{array}{c} \sum_{u \in S_1 } \Psi_u \\ \vdots \\ \sum_{u \in S_m } \Psi_u \end{array} \right] ,
$$
where the matrix $A$ corresponds to the pooling operation. 
This means of compression corresponds to ``mixing'' measurements at randomly chosen transform domain values $u$. (When $\Psi$ is the Fourier transform, this corresponds to randomly chosen frequencies, and when $\Psi$ is the wavelet, this corresponds to mixing randomly chosen translation and scale parameters) .  When the signal $F$ is acquired in this manner, we obtain measurements of the form:
$$
Y=A \hat{F} A^T,
$$
where $\hat{F}$ is suitably sparse. Note that one may choose different random mixtures of measurements for the $t$ and $r$ ``spatial'' variables, in which case one would obtain measurements of the form:
$$
Y=A \hat{F} B^T.
$$
The theory developed in this paper shows how one can recover the multi-dimensional signal $F$ from such an undersampled acquisition mechanism. In particular, our results will show that if the pooling of the transform variable is done suitably randomly, then there is an efficient method based on linear programming that can be used to recover the original multi-dimensional signal. 


\subsection{Related Work and Obstacles to Common Approaches}
\label{sec.relatedWork}
The problem of recovering sparse signals via $\ell_1$ regularization and convex optimization has been studied extensively in the past decade; our work fits broadly into this context. In the signal processing community, the literature on compressed sensing \cite{candes2006robust,donoho2006compressed} focuses on recovering sparse signals from data.  In the statistics community, the LASSO formulation as proposed by Tibshirani, and subsequently analyzed (for variable selection) by Meinshausen and B\"{u}hlmann \cite{meinshausenBuhlmann2006}, and Wainwright \cite{wainwright2009sharp} are also closely related. Other examples of
structured model selection include estimation of models with a few latent factors (leading to low-rank covariance matrices) \cite{Fan_Fan_Lv_2007}, models specified by banded or sparse covariance matrices \cite{BickelLevina2,BickelLevina1}, and Markov or graphical models \cite{Lauritzen_book,meinshausenBuhlmann2006,RWRY}. These ideas have been studied in depth and extended to analyze numerous other model selection problems in statistics and signal processing \cite{Recht,venkatConvexFOCM2012,RobustPCA}. 

Our work is also motivated by the work on sketching in the computer science community; this literature deals with the idea of compressing high-dimensional data vectors via projection to low-dimensions while preserving pertinent geometric properties. The celebrated Johnson-Lindenstrauss Lemma \cite{johnson1984extensions} is one such result, and the idea of sketching has been explored in various contexts \cite{Andoni09efficientsketches,Kane:2011:FME:1993636.1993735}. The idea of using random bipartite graphs and their related expansion properties, which motivated our approach to the problem, have also been studied in past work \cite{berindeIndyk2008,journals/corr/abs-0902-4045,khajehnejad2011sparse}.

While most of the work on sparse recovery focuses on sensing matrices where each entry is an i.i.d. random variable, there are a few lines of work that explore structured sensing matrices. For instance, there have been studies of matrices with Toeplitz structure \cite{haupt2010toeplitz}, or those with random entries with independent rows but with possibly dependent columns \cite{vershynin2010introduction,rudelson2008sparse}. Also related is the work on deterministic dictionaries for compressed sensing \cite{cormode2006combinatorial}, although those approaches yield results that are too weak for our setup. 

One interesting aspect of our work is that we show that it is possible to use highly constrained sensing matrices (i.e. those with tensor product structure) to recover the signal of interest. Many standard techniques fail in this setting.  Restricted isometry based approaches \cite{candes2005decoding} and coherence based approaches \cite{donoho2003optimally,gribonval2003sparse,tropp2004greed} fail due to a lack of independence structure in the sensing matrix. Indeed, the restricted isometry constants as well as the coherence constants are known to be weak for tensor product sensing operators \cite{duarteKroneckerCS2010,jokarKroneckerCS2010}. Gaussian width based analysis approaches \cite{venkatCnvxGeometryArxiv2010} fail because the kernel of the sensing matrix is not a uniformly random subspace and hence not amenable to a similar application of Gordon's (``escape through the mesh'') theorem. We overcome these technical difficulties by working directly with combinatorial properties of the tensor product of a random bipartite graph, and exploiting those to prove the so-called nullspace property \cite{Donoho2001,Cohen2009}.
\section{Experiments}
\label{sec.experiments}
We demonstrate the validity of our theory with some preliminary experiments in this section. Figure~\ref{fig.recoveryExample} shows a $40\times 40$ distributed sparse matrix on the left side. The matrix on the right is a perfect reconstruction using a sketch dimension of $m=21$. 
\begin{figure}[!t]
\centering
\includegraphics[scale=0.7]{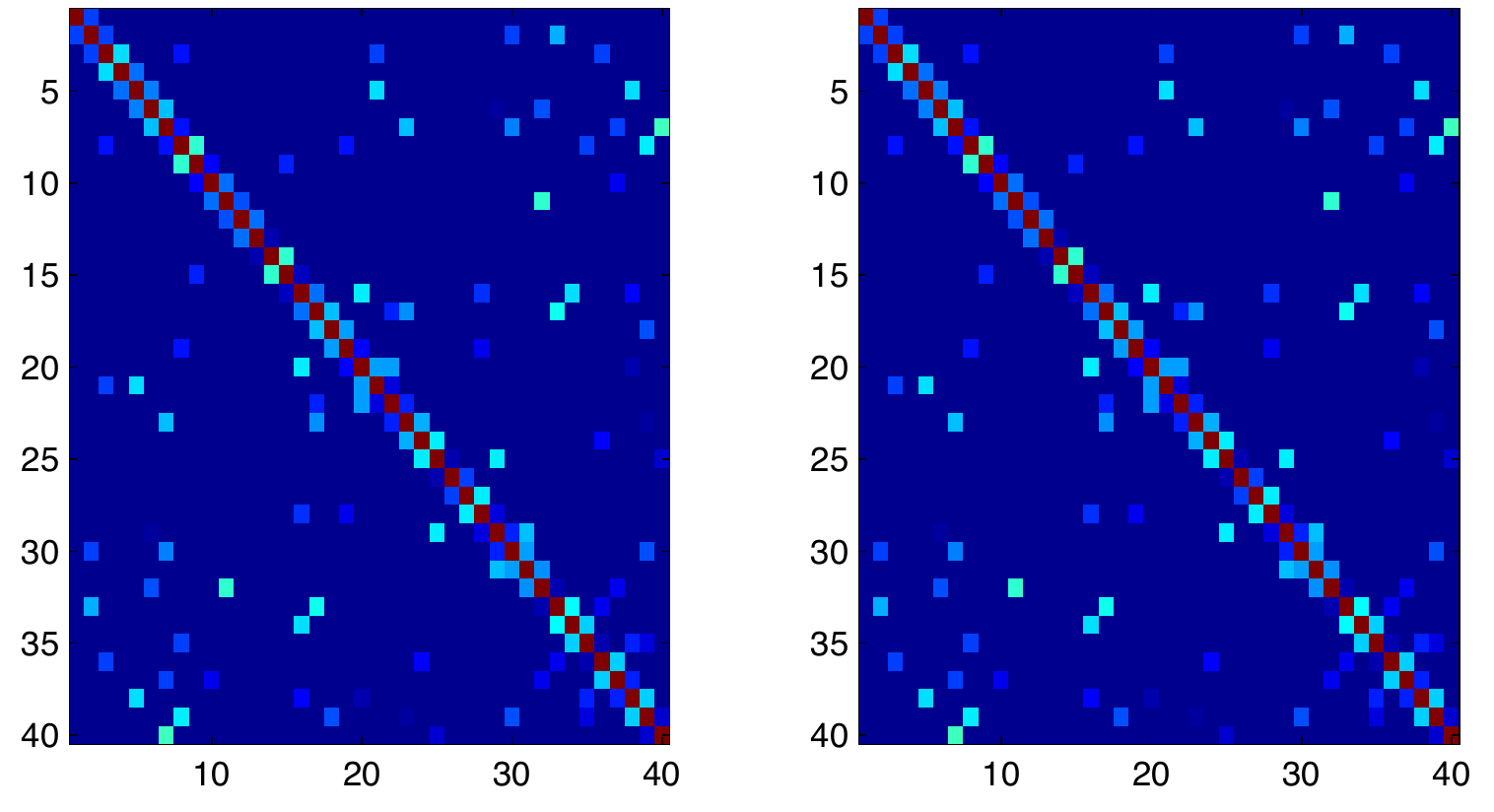}
\caption{The matrix on the left is a $40\times 40$ sparse matrix and the matrix on the right is a perfect reconstruction with $m = 21$. }
\label{fig.recoveryExample}
\end{figure}
\begin{figure}[!b]
\centering
\includegraphics[scale=0.45]{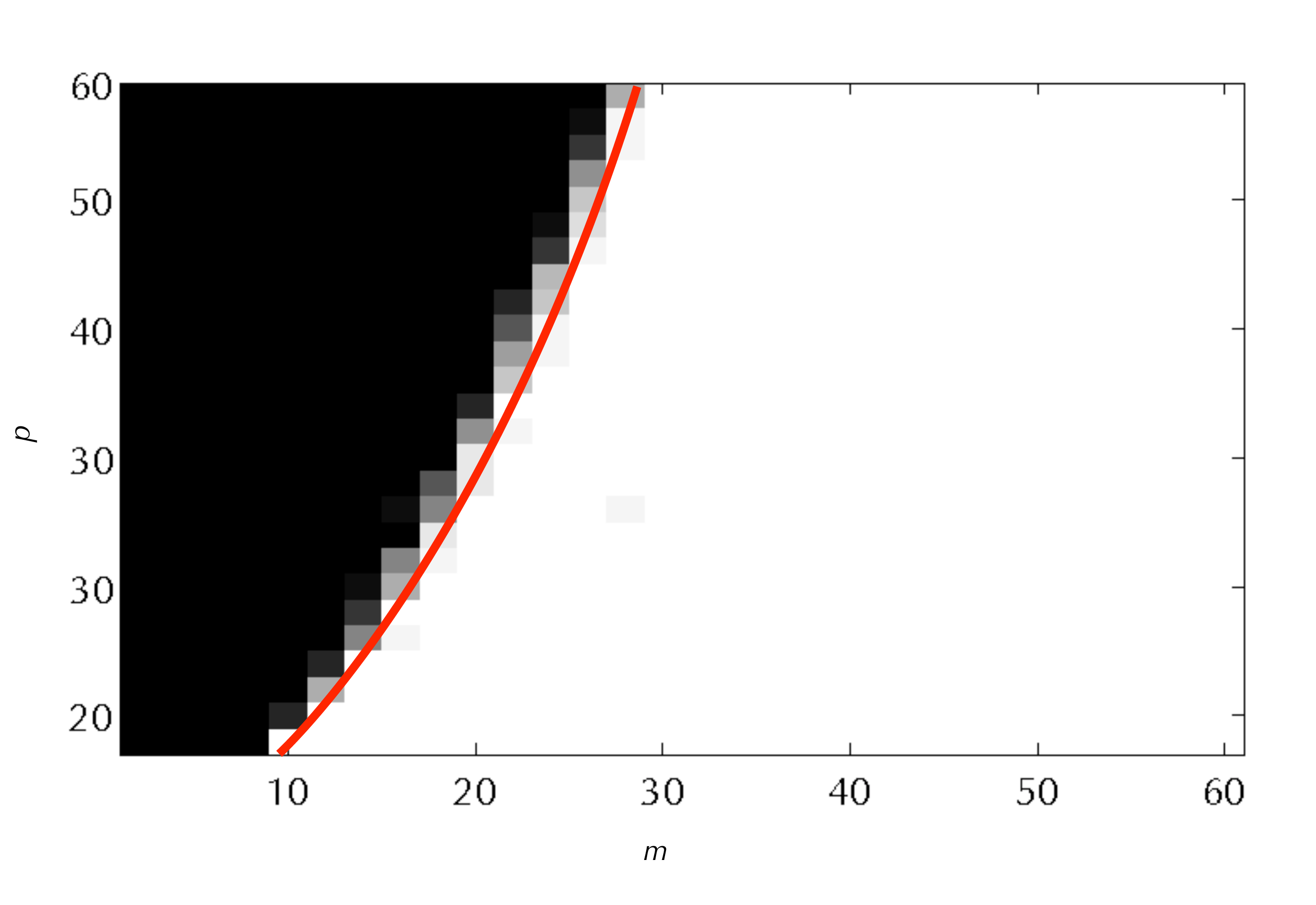}
\caption{Phase transition plot. The $(i,j)$-th pixel shows (an approximation) to the probability of success of the optimization problem \eqref{eq.optP1} in recovering a distributed sparse $X\in\mathbb{R}^{i\times i}$ with sketch-size $j$. The (red) solid line shows the boundary of the phase transition regime and is approximately the curve $p=\frac{1}{14}m^2$}
\label{fig.phaseTransition}
\end{figure}

\begin{figure}[!t]
\centering
\includegraphics[scale=0.5]{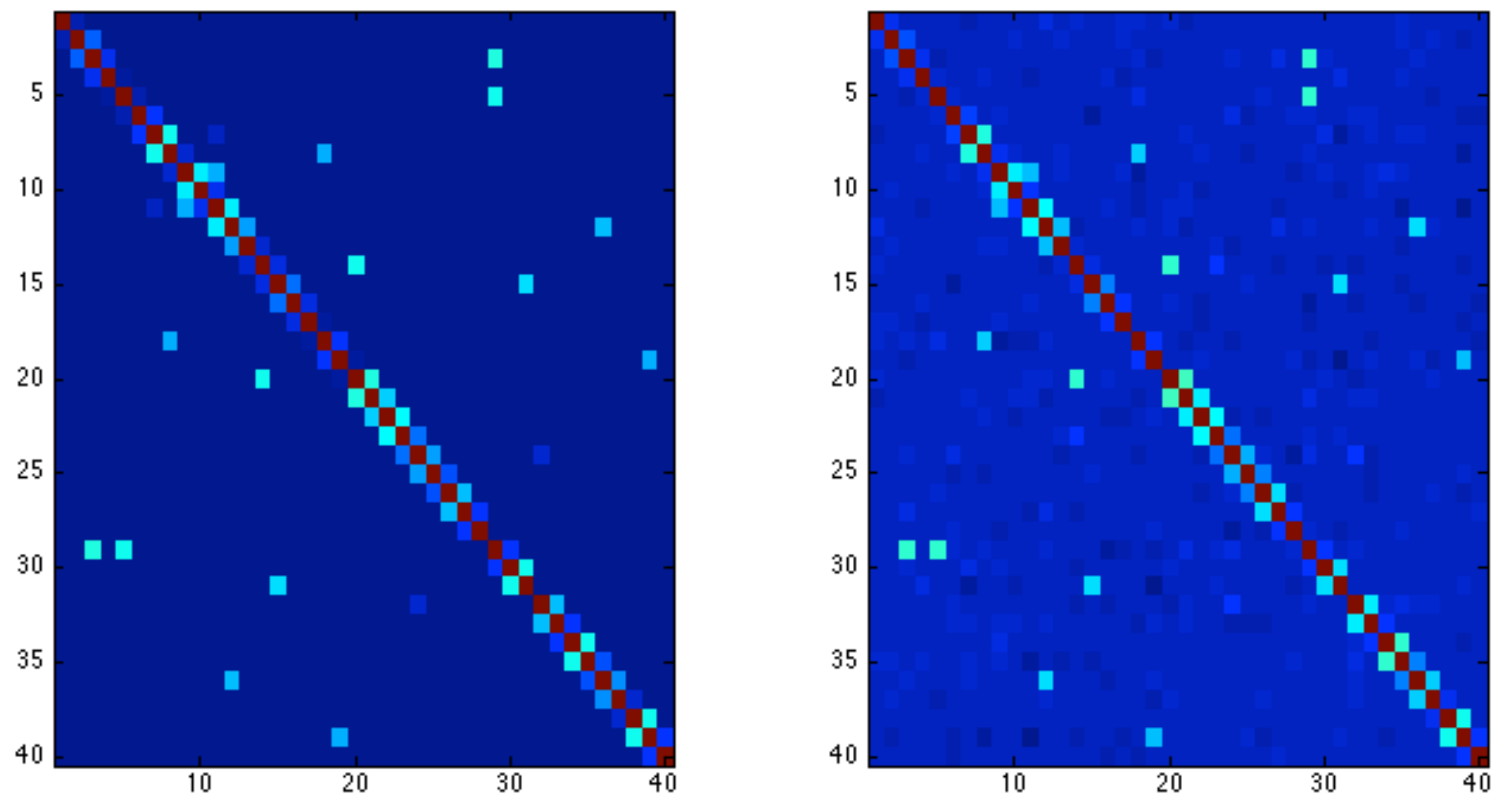}
\caption{The matrix on the left is a $40\times 40$ distributed sparse matrix. The matrix on the right was reconstructed using $n=2100$ samples and with sketches of size $m=21$.}
\label{fig.approximateReconstruction}
\end{figure}

Figure~\ref{fig.phaseTransition} is what is known now as the ``phase transition diagram''. Each coordinate $(i,j)\in\{10,12,\dots,60\}\times\left\{2,4,\ldots,60\right\}$ in the figure corresponds to an experiment with $p=i$ and $m=j$. The value at the coordinate $(i,j)$ was generated as follows. A random 4-distributed sparse $X\in\mathbb{R}^{i\times i}$ was generated and a random $A\in\mathbb{R}^{j\times i}$ was generated as {adjacency matrix of a graph as described in Definition~\ref{def.randomGraph}}. Then the optimization problem \eqref{eq.optP1} was solved using the CVX toolbox \cite{cvx,gb08}. The solution $X^\ast$ was compared to $X$ in the $\left\|\cdot\right\|_\infty$ norm (upto numerical precision errors). This was repeated 40 times and the average number of successes was reported in the $(i,j)$-th spot. In the figure, the  white region denotes success during each trial and the black region denotes failure in every single trial and there is a sharp \emph{phase transition} between successes and failures. In fact, the curve that borders this phase transition region roughly looks like the curve $p = \frac{1}{14}m^2$ which is what our theory predicts (upto constants and log factors). 

We also ran some preliminary tests on trying to reconstruct a covariance matrix from sketches of samples drawn from the original distribution.To factor in the ``noise'', we replaced the equality constraint in \eqref{eq.optP1} with a constraint which restricts the feasible set to be the set of all $X$ such that $\left\|AXA^T-\hat{X}^{(n)}_Y\right\|_2\leq\kappa$ instead. The parameter $\kappa$ was picked by cross-validation. Figure~\ref{fig.approximateReconstruction} shows a representative result which is encouraging. The matrix on the left is a $40\times 40$ distributed sparse covariance matrix and the matrix on the right is a reconstruction using $n = 2100$ sketches of size $m=21$ each.


\section{Preliminaries and Notation}
\label{sec.preliminaries}

We will begin our theoretical discussion by establishing some notation and preliminary concepts that will be used through the rest of the paper.
 
For any $p\in\mathbb{N}$, we define $[p]:=\left\{1,2,\ldots,p\right\}$. A {\bf graph} $G=(V,E)$ is defined in the usual sense as an ordered pair with the \emph{vertex set} $V$ and the \emph{edge set} $E$ which is a set of 2-element subsets of $V$, i.e., $E\subset {V\choose 2}$. Henceforth, unless otherwise stated, we will deal with the graph $G=([p],E)$. We also assume that all the graphs that we consider here include all the self loops, i.e., $\left\{i,i\right\}\in E$ for all $i\in [p]$. For any $S\subset[p]$, the set of neighbors $N(S)$ is defined as 
$$N(S) = \left\{j\in [p]: i\in S,\left\{i,j\right\}\in E \right\}.$$
For any vertex $i\in[p]$, the degree deg$(i)$ is defined as deg$(i) := \left|N(i)\right|$. 
\begin{definition}[Bounded degree graphs and regular graphs]
A graph $G=([p],E)$ is said to be a {\bf bounded degree graph} with (maximum) degree $d$ if for all $i\in [p]$, $$\mbox{deg}(i)\leq d$$
The graph is said to be {\bf $d-$regular} if deg$(i) = d$ for all $i\in [p]$. 
\end{definition}

We will be interested in another closely related combinatorial object. Given $p,m\in\mathbb{N}$, a {\bf bipartite graph} $G=([p],[m],E)$ is a graph with the \emph{left set} $[p]$ and \emph{right set} $[m]$ such that the edge set $E$ only has pairs $\left\{i,j\right\}$ where $i$ is the left set and $j$ is in the right set. A bipartite graph $G=([p],[m],E)$ is said to be {\bf $\delta-$left regular} if for all $i$ in the left set $[p]$, deg$(i) = \delta$. Given two sets $A\subset [p], B\subset [m]$, we define the set $$E\left(A:B\right):=\left\{(i,j)\in E: i\in A, j\in B\right\},$$
which we will find use for in our analysis. This set is sometimes known as the \emph{cut set}. Finally for a set $A\subset[p]$ (resp. $B\subset [m]$), we define $N(A) :=\{j\in[m]: i\in A, \{i,j\}\in E\}$ (resp.  $N_R(B) :=\{i\in[p]: j\in B, \{i,j\}\in E\}$ ). This distinction between $N$ and $N_R$ is made only to reinforce the meaning of the quantities which is otherwise clear in context.

\begin{definition}(Tensor graphs)
Given two bipartite graphs $G_1 = ([p],[m],E_1)$ and $G_2 = ([p],[m],E_2)$, we define their {\bf tensor graph} $G_1\otimes G_2$ to be the bipartite graph $([p]\times [p], [m]\times [m], E_1\otimes E_2)$ where $E_1\otimes E_2$ is such that $\left\{(i,i'),(j,j')\right\}\in E_1\otimes E_2$ if and only if $\left\{i,j\right\}\in E_1$ and $\left\{i',j'\right\}\in E_2$. 
\end{definition}
Notice that if the adjacency matrices of $G_1, G_2$ are given respectively by $A^T,B^T\in \mathbb{R}^{p\times m}$, then the adjacency matrix of $G_1\otimes G_2$ is $\left(A\otimes B\right)^T\in\mathbb{R}^{p^2\times m^2}$.

As mentioned earlier, we will be particularly interested in the situation where $B = A$. In this case, write the tensor product of a graph $G = ([p],[m],E)$ with itself as $G^\otimes = ([p]\times [p], [m]\times [m], E^\otimes)$. Here $E^\otimes$ is such that $\left\{(i,i'),(j,j')\right\}\in E^\otimes$ if and only if $\left\{i,j\right\}$ and $\left\{i',j'\right\}$ are in $E$. 
\\

Throughout this paper, we write $\|\cdot\|$ to denote norms of vectors. For instance, $\|x\|_1$ and $\|x\|_2$ respectively stand for the $\ell_1$ and $\ell_2$ norm of $x$. Furthermore, for a matrix $X$, we will often write $\|X\|$ to denote $\|\mbox{vec} (X)\|$  to avoid clutter. Therefore, the Frobenius norm of a matrix $X$ will appear in this paper as $\|X\|_2$.

\subsection{ Distributed Sparsity}
\label{sec.distributedSparisity}
As promised, we will now argue that distributed sparsity is important. Towards this end, let us turn our attention to Figure~\ref{fig.arrowVsDistributed} which shows two matrices with $\mathcal{O}(p)$ non-zeros. Suppose that the non-zero pattern in $X$ looks like that of the matrix on the left (which we dub as the ``arrow'' matrix). It is clear that it is impossible to recover this $X$ from $AX B^T$ \emph{even if} we know the non-zero pattern in advance. 
\begin{figure}[t!]
\centering
\includegraphics[scale=1.1]{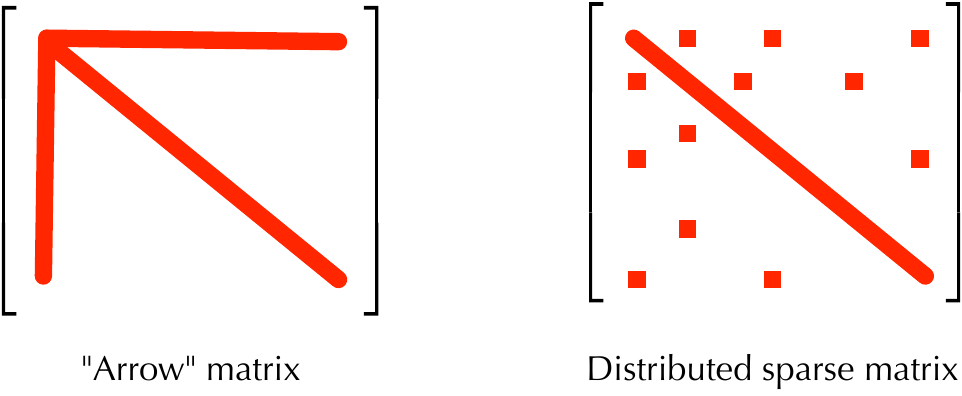}
\caption{Two matrices with $\mathcal{O}(p)$ non-zeros. The ``arrow'' matrix is impossible to recover by covariance sketching while the distributed sparse matrix is. }
\label{fig.arrowVsDistributed}
\end{figure}
For instance, if $v\in \mbox{ker}(A)$, then the matrix $\tilde{X}$, with $v$ added to the first column of $X$
is such that $AX B^T = A\tilde{X}B^T$ and $\tilde{X}$ is also an arrow matrix and hence indistinguishable from $X$. Similarly, one can ``hide'' a kernel vector of $B$ in the first row of the arrow matrix. In other words, it is impossible to uniquely recover $X$ from $A X B^T$. 

In what follows, we will show that if the sparsity pattern of $X$ is more \emph{distributed}, as in the right side of Figure~\ref{fig.arrowVsDistributed}, then one can recover $X$ and do so efficiently. In fact, our analysis will also reveal what that the size of the sketch $Y$ needs to be able to perform this task and we will see that this is very close to being optimal. 
\\
In order to make things concrete, we will now define these notions formally. \\
\begin{definition}[$d-$distributed sets and $d-$distributed sparse matrices]
\label{def.distributedSparse}
We say that a subset $\Omega\subset [p]\times [p]$ is {\bf $d-$distributed }if the following hold.
\begin{enumerate}
\item For $i = 1,2,\ldots,p$, $(i,i)\in \Omega$.
\item For all $k\in [p]$, the cardinality of the sets $\Omega_{k}:= \left\{(i,j)\in\Omega: i = k\right\}$ and $\Omega^{k}:= \left\{(i,j)\in\Omega: j = k\right\}$ is no more than $d$.
\end{enumerate}
The set of all $d-$distributed subsets of $[p]\times[p]$ will be denoted as $\mathfrak{W}_{p,d}$. We say that a matrix $X\in\mathbb{R}^{p\times p}$ is {\bf $d-$distributed sparse} if there exists an $\Omega\in\mathfrak{W}_{d,p}$ such that supp$(X):=\left\{(i,j)\in [p]\times [p]: X_{ij}\neq 0\right\}\subset \Omega$.
\end{definition}

While the theory we develop here is more generally applicable, the first point in the above definition makes the presentation easier. Notice that this forces the number of off-diagonal non-zeros in each row or column of a $d-$distributed sparse matrix $X$ to be at most $d-1$. This is not a serious limitation since a more careful analysis along these lines can only improve the bounds we obtain by at most a constant factor. 

\noindent{\bf Examples}:
\begin{itemize*}
\item Any diagonal matrix is $d-$distributed sparse for $d= 1$. Similarly a tridiagonal matrix is $d-$distributed sparse with $d=3$. 
\item The adjacency matrix of a bounded degree graph with maximum degree $d-1$ is $d-$distributed sparse
\item As shown in Proposition~\ref{prop.random_sparsity}, \emph{random sparse matrices} are $d-$distributed sparse with $d = \mathcal{O}(\log p)$. While we implicitly assume that $d$ is constant with respect to $p$ in what follows, all arguments work even if $d$ grows logarithmically in $p$ as is the case here. 
\end{itemize*} 

Given a matrix $X\in\mathbb{R}^{p\times p}$, as mentioned earlier, we write vec$(X)$ to denote the $\mathbb{R}^{p^2}$ vector obtained by stacking the columns of $X$. Suppose $x = \mbox{vec}(X)$. It will be useful in what follows to remember that $x$ was actually derived from the matrix $X$ and hence we will employ slight abuses of notation as follows: We will say $x$ is $d-$distributed sparse when we actually mean that the matrix $X$ is. Also, we will write $(i,j)$ to denote the index of $x$ corresponding to $X_{ij}$, i.e., 
$$x_{ij}:=x_{(i-1)p+j} = X_{ij}.$$

\noindent We finally make two remarks: 
\begin{itemize}
\item Even if $X$ were distributed sparse, the observed vector $Y = AX B^T$ is usually unstructured (and dense). 
\item While the results in this paper focus on the regime where the maximum number of non-zeros $d$ per row/column of $X$ is a constant (with respect to $p$), they can be readily extended to the case when $d$ grows poly-logarithmically in $p$. Extensions to more general scalings of $d$ is an interesting avenue for future work.
\end{itemize}


\subsection{Random Bipartite Graphs, Weak Distributed Expansion and the Choice of the Sketching Matrices}
\label{sec.mainLemma}
As alluded to earlier, we will choose the sensing matrices $A,B$ to be the adjacency matrices of certain random bipartite graphs. The precise definition of this notion follows. 
\begin{definition}[Uniformly Random $\delta-$left regular bipartite graph]
\label{def.randomGraph}
We say that $G = ([p],[m],E)$ is a { uniformly random $\delta-$left regular bipartite graph} if the edge set $E$ is a random variable with the following property: for each $i\in [p]$ one chooses $\delta$ vertices $j_1,j_2,\ldots,j_{\delta}$ chosen uniformly and independently at random (with replacement) from $[m]$ such that $\left\{\{i,j_k\}\right\}_{k=1}^\delta\subset E$.
\end{definition}
\noindent{\bf Remarks:\vspace{-2mm}}
\begin{itemize} 
\item Note that since we are sampling with replacement, it follows that the bipartite graph thus constructed may not be simple. If for instance there are two edges from the left node $i$ to the right node $i$, the corresponding entry $A_{ij}=2$.
\item It is in fact possible to work with a sampling without replacement model in Definition \ref{def.randomGraph} (where the resulting graph is indeed simple) and obtain qualititatively the same results. We work with a ``sampling with replacement'' model for the ease of exposition.
\end{itemize}
The probabilistic claims in this paper are made with respect to this probability distribution on the space of all bipartite graphs.

In past work \cite{berindeIndyk2008,journals/corr/abs-0902-4045}, the authors show that a random graph generated as above is, for suitable values of $\epsilon,\delta$, a $(k,\delta,\epsilon)-$expander. That is, for all sets $S\subset[p]$ such that $\left|S\right|\leq k$, the size of the neighborhood $\left|N(S)\right|$ is no less than $(1-\epsilon)\delta\left|S\right|$.  If $A$ is the adjacency matrix of such a graph, then it can then be shown that this implies that $\ell_1$ minimization would recover a $k-$sparse vector $x$ if one observes the sketch $Ax$ (actually, \cite{berindeIndyk2008} shows that these two properties are equivalent). Notice that in our context, the vector that we need to recover is $\mathcal{O}(p)$ sparse and therefore, our random graph needs to be a $(\mathcal{O}(p),\delta,\epsilon)-$expander. Unfortunately, this turns out to not be true of $G_1\otimes G_2$ when $G_1$ and $G_2$ are randomly chosen as directed above. 

However, we prove that if $G_1$ and $G_2$ are picked as in Definition~\ref{def.randomGraph}, then their tensor graph  $G_1\otimes G_2$, satisfies what can be considered a \emph{weak distributed expansion} property. This roughly says that the neighborhood of a $d-$distributed $\Omega\subset[p]\times [p]$ is large enough. Moreover, we show that this is in fact sufficient to prove that with high probability $X$ can be recovered from $AX B^T$  efficiently. The precise statement of these combinatorial claims follows. \\
\begin{lemma}
\label{lemma.main}
Suppose that $G_1 = ([p],[m],E_1)$ and $G_2 = ([p],[m],E_2)$ are two independent uniformly random $\delta-$left regular bipartite graphs with  $\delta = \mathcal{O}(\log p)$ and $m=\mathcal{O}(\sqrt{dp} \log p)$. Let $\Omega\in\mathfrak{W}_{d,p}$ be fixed. Then there exists an $\epsilon \in \left( 0,\frac{1}{4} \right)$ such that $G_1\otimes G_2$ has the following properties with probability exceeding $1-p^{-c}$, for some $c>0$. 
\begin{enumerate}
\item  $|N(\Omega)| \geq p \delta^2(1-\epsilon)$.
\item For any $(i,i') \in ([p] \times [p]) \setminus \Omega $ we have $|N(i,i')\cap N(\Omega)| \leq \epsilon \delta^2$.  
\item For any $(i,i') \in \Omega$, $|N(i,i')\cap N(\Omega \setminus (i,i'))| \leq \epsilon \delta^2$.  
\end{enumerate}
Moreover, all these claims continue to hold when $G_2$ is the same as $G_1$. 	
\end{lemma}

\noindent{\bf Remarks:\vspace{-2mm}}
\begin{itemize}
\item Part 1 of Lemma~\ref{lemma.main} says that if $\Omega$ is a $d-$distributed set, then the size of the neghborhood of $\Omega$ is large. This can be considered a \emph{weak distributed expansion property}. Notice that while it is reminiscent of the \emph{vertex expansion} property of expander graphs, it is easy to see that it does not hold if $\Omega$ is not distributed. Furthermore, we call it ``weak'' because the lower bound on the size of the neighborhood is $\delta^2p(1-\epsilon)$ as opposed to $\delta^2\left|\Omega\right|(1-\epsilon) = \delta^2dp(1-\epsilon)$ as one typically gets for standard expander graphs. It will become clear that this is one of the key combinatorial facts that ensures that the necessary theoretical guarantees hold for covariance sketching.

\item Parts 2 and 3 say that the number of collisions between the edges emanating out of a single vertex with the edges emanating out of a distributed set is small. Again, this combinatorial property is crucial for the proof of Theorem~\ref{thm.main} to work. 
\end{itemize}

As stated earlier, we are particularly interested in the challenging case when $G_1 = G_2$ (or equivalently, their adjacency matrices $A,B$ are the same). The difficultly, loosely speaking, stems from the fact that since we are not allowed to pick $G_1$ and $G_2$ separately, we have much less independence. In Appendix~\ref{appendix:A}, we will only prove Lemma~\ref{lemma.main} in the case when $G_1=G_2$ since this proof can be modified in a straightforward manner to obtain the proof of the case when $G_1$ and $G_2$ are drawn independently.


\section{Proofs of Main Results}
\label{sec.proof}
In this section, we will prove the main theorem. To reduce clutter in our presentation, we will sometimes employ certain notational shortcuts. When the context is clear, the ordered pair $(i,i')$ will simply be written as $ii'$ and the set $[p]\times [p]$ will be written as $[p]^2$.  Sometimes, we will also write $\cA$ to mean $A\otimes A$ and if $S\subset[p]\times [p]$, we write $(A\otimes A)_S$ or $\cA_S$ to mean the submatrix of $A\otimes A$ obtained by appending the columns $\{A_i\otimes A_j\mid (i,j)\in S\}$.

As stated earlier, we will only provide the proof of Theorem~\ref{thm.main} for the case when $A=B$. Some straightforward changes to the proof presented here readily gives one the proof for the case when the matrices $A$ and $B$ are distinct. 
\subsection{Proof of Theorem~\ref{thm.main}}

We will consider an arbitrary ordering of the set $[p]\times [p]$ and we will order the edges in $E^\otimes$ lexicographically based on this ordering, i.e., the first $\delta^2$ edges $e_1,\ldots,e_{\delta^2}$ in $E^\otimes$ are those that correspond to the first element as per the ordering on $[p]\times [p]$ and so on. Now, one can imagine that the graph $G^\otimes$ is formed by including these edges sequentially as per the ordering on the edges. This allows us to partition the edge set into the set $E_1^\otimes$ of edges that do not collide with any of the previous edges as per the ordering and the set $E_2^\otimes:=E^\otimes - E_1^\otimes$. (We note that a similar proof technique was adopted in Berinde et al. \cite{berindeIndyk2008}). 

As a first step towards proving the main theorem, we will show that the operator $A\otimes A$ preserves the $\ell_1$ norm of a matrix $X$ as long as $X$ is distributed sparse. Berinde et al.,\cite{berindeIndyk2008} call a similar property RIP-1, taking cues from the restricted isometry property that has become popular in literature \cite{candes2005decoding}. The proposition below can also be considered to be a restricted isometry property but the operator in our case is only constrained to behave like an isometry for distributed sparse vectors. 
\begin{proposition}[$\ell_1$-RIP]
\label{prop.l1rip}
Suppose $X\in\mathbb{R}^{p\times p}$ is $d-$distributed sparse and  $A$ is the adjacency matrix of a random bipartite $\delta-$left regular graph. Then there exists an $\epsilon>0$ such that
\begin{equation}
(1-2\epsilon)\delta^2\left\|X\right\|_1\leq\left\|AXA^T\right\|_1\leq \delta^2\left\|X\right\|_1,
\end{equation}
with probability exceeding $1-p^{-c}$ for some $c>0$.
\end{proposition}
\begin{proof}
The upper bound follows (deterministically) from the fact that the induced (matrix) $\ell_1$-norm of $A\otimes A$, i.e., the maximum column sum of $A\otimes A$, is precisely $\delta^2$. To prove the lower bound, we need the following lemma.
\begin{lemma}
\label{lemma.towardsL1rip}
For any $X\in\mathbb{R}^{p\times p}$, 
\begin{align}
\left\|AXA^T\right\|_1\geq\delta^2\left\|X\right\|_1 - 2\sum_{jj'\in[m]^2}\sum_{ii'\in [p]^2}{\mathbf{1}_{\{ii',jj'\}\in E^\otimes_2}}|X_{ii'}|\label{eq.lemma1}
\end{align}
\end{lemma}
\begin{proof} In what follows, we will denote the indicator function $\mathbf{1}_{\{ii',jj'\}\in S}$ by $\mathbf{1}_S^{\{ii'jj'\}}$. We begin by observing that
\begin{align}
   \left\| AXA^T \right\|_1
   & = \left\|\mathcal{A}\,\mbox{vec}(X)\right\|_1\notag\\
   & = \sum_{j j'\in[m]^2} \left| \sum_{ii'\in[p]^2} \mathcal{A}_{\left\{ii'jj'\right\}}X_{i i'} \right|\notag\\
   & = \sum_{j j'\in[m]^2} \left| \sum_{ii'\in[p]^2}  \mathbf{1}^{\{ii'jj'\}}_{E^\otimes}\hspace{1mm}X_{i i'} \right|\notag\\
 & =  \sum_{jj'\in[m]^2} \left| \sum_{ii'\in[p]^2}\mathbf{1}^{\{ii'jj'\}}_{{E}^{\otimes}_1} X_{i i'} +\sum_{ii'\in[p]^2}\mathbf{1}^{\{ii'jj'\}}_{{E}^{\otimes}_2} X_{i i'}\right|\notag\\
& \geq \sum_{jj'\in[m]^2} \left| \sum_{ii'\in[p]^2}\mathbf{1}^{\{ii'jj'\}}_{{E}^\otimes_1} X_{i i'} \right| - \left| \sum_{ii'\in[p]^2}\mathbf{1}^{\{ii'jj'\}}_{{E}^{\otimes}_2}X_{i i'}\right|\notag\\
 & \stackrel{(a)}{\geq}  \sum_{jj'\in[m]^2}\Big( \sum_{ii'\in[p]^2}\mathbf{1}^{\{ii'jj'\}}_{{E}^\otimes_1}\left| X_{i i'} \right|- \sum_{ii'\in[p]^2}\mathbf{1}^{\{ii'jj'\}}_{E^{\otimes}_2} \left| X_{i i'} \right|\Big)\notag\\
 & = \sum_{ii'\in[p]^2,jj'\in[m]^2} \mathbf{1}^{\{ii'jj'\}}_{E^\otimes} \left| X_{i i'} \right|-2\sum_{ii'\in[p]^2jj'\in[m]^2}\mathbf{1}^{\{ii'jj'\}}_{{E}^{\otimes}_2} \left| X_{i i'} \right|\notag,
\end{align}
where $(a)$ follows after observing that the first (double) sum has only one term and applying triangle inequality to the second sum. Since \\$\sum_{ii'\in[p]^2,jj'\in[m]^2} \mathbf{1}^{\{ii'jj'\}}_{E^\otimes}\left| X_{i i'} \right| = \delta^2\left\|X\right\|_1$, this concludes the proof of the lemma. 
\end{proof}
Now, to complete the proof of Proposition \ref{prop.l1rip}, we need to bound the sum in the LHS of \eqref{eq.lemma1}. 
Notice that $$\sum_{ii'jj': \left( i i' j j'\right) \in{E}^{\otimes}_2} \left| X_{i i'} \right| = \sum_{ii'}\left|X_{ii'}\right|r_{ii'} = \sum_{ii'\in\Omega}\left|X_{ii'}\right|r_{ii'}.$$
{where $r_{ii'}$ is the number of collisions of edges emanating from $ii'$ with all the previous edges as per the ordering we defined earlier}. Since $\Omega$ is $d-$distributed, from the third part of Lemma \ref{lemma.main}, we have that for all $ii'\in\Omega$, $r_{ii'}\leq \epsilon\delta^2$ with probability exceeding $1-p^{-c}$ and therefore, 
$$\sum_{ii'\in\Omega}\left|X_{ii'}\right|r_{ii'}\leq\epsilon\delta^2\left\|X\right\|_1.$$
This concludes the proof. 
\end{proof}
Next, we will use the fact that $A\otimes A$ behaves as an approximate isometry (in the $\ell_1$ norm) to prove what can be considered a nullspace property \cite{Donoho2001,Cohen2009}. This will tell us that the nullspace of $A\otimes A$ is ``smooth'' with respect to distributed support sets and hence $\ell_1$ minimization as proposed in \eqref{eq.optP1} will find the right solution.

\begin{proposition}[Nullspace Property]
\label{prop.nsp}
Suppose that $A \in \left\{0,1\right\}^{m \times p}$ is the adjacency matrix of a random bipartite $\delta-$left regular graph with $\delta=\mathcal{O}(\log p)$ and ${m=\mathcal{O}(\sqrt{dp}\log p)}$ and that $\Omega\in\mathfrak{W}_{d,p}$ is fixed. Then, with probability exceeding $1-p^{-c}$, for any $V\in \mathbb{R}^{p\times p}$ such that $AVA^T = 0$, we have
\begin{equation}
\left\|V_\Omega\right\|_1\leq {\frac{\epsilon}{1-3\epsilon}}\left\|V_{\Omega^c}\right\|_1.
\end{equation}
for some $\epsilon\in (0,\frac{1}{4})$ and for some $c>0$. 
\end{proposition}

\begin{proof}
Let $V$ be any symmetric matrix such that $AVA^T=0$. Let ${v=\text{vec(V)}}$ and note that vec$\left(AVA^T\right)=\left(A \otimes A\right) v=0$. Let $\Omega$ be a $d$-distributed set. As indicated in Section~\ref{sec.preliminaries}, we define $N(\Omega) \subseteq [m]^2$ to be the set of neighbors of $\Omega$ with respect to the graph $G^\otimes$. Let $(A\otimes A)^{N(\Omega)}$ denote the submatrix of $A\otimes A$ that contains only those rows corresponding to $N(\Omega)$ (and all columns).  We will slightly abuse notation and use $v_\Omega$ to denote the vectorization of the projection of $V$ onto the set $\Omega$, i.e., $v_\Omega =\mbox{vec}(V_\Omega)$, where $$[V_\Omega]_{i,j}=\begin{cases}
V_{i,j} & (i,j)\in\Omega\\
0 & \mbox{otherwise}
\end{cases}.$$ Now, we can follow the following chain of inequalities:
\begin{align}
0&=\left\|(A\otimes A)^{N(\Omega)} v\right\|_1 \notag\\
&= \left\|(A\otimes A)^{N(\Omega)} (v_\Omega + v_{\Omega^c})\right\|_1 \notag\\
&\geq \left\|(A\otimes A)^{N(\Omega)} v_\Omega\right\|_1 -\left\|(A\otimes A)^{N(\Omega)} v_{\Omega^c}\right\|_1 \notag\\
& = \left\|(A \otimes A) v_\Omega\right\|_1 -\left\|(A\otimes A)^{N(\Omega)} v_{\Omega^c}\right\|_1\notag\\
& \geq (1-2\epsilon) \delta^2 \left\|V_\Omega\right\|_1 - \left\|(A\otimes A)^{N(\Omega)} v_{\Omega^c}\right\|_1\notag,
\end{align}
where the last inequality follows from Proposition \ref{prop.l1rip}.
Resuming the chain of inequalities, we have:
\begin{align}
0 &\geq (1-2\epsilon) \delta^2 \left\|V_\Omega\right\|_1 - \sum_{ii'\in\Omega^c}\left\|(A \otimes A)^{N(\Omega)} v_{\{ii'\}}\right\|_1 \notag\\
&\geq  (1-2\epsilon) \delta^2 \left\|V_\Omega\right\|_1 - \sum_{\substack{ii',jj': (ii',jj') \in E^\otimes, \\ jj' \in N(\Omega),ii'\in\Omega^c}} |V_{ii'}| \notag\\
& =   (1-2\epsilon) \delta^2 \left\|V_\Omega\right\|_1 - \sum_{ii'\in\Omega^c}{|E^\otimes (ii':N(\Omega))|}\; |V_{ii'}| \notag\\
& \stackrel{(a)}{\geq}   (1-2\epsilon) \delta^2 \left\|V_\Omega\right\|_1 - \sum_{ii'\in\Omega^c}{\epsilon  \delta^2 }\; \left|V_{ii'}\right|\notag \\
& \geq   (1-2\epsilon) \delta^2 \left\|V_\Omega\right\|_1 -  {\epsilon} \delta^2\left\|V\right\|_1\notag,
\end{align}
where $(a)$ follows from the second part of Lemma \ref{lemma.main}. Writing $\left\|V\right\|_1=\left\|V_{\Omega}\right\|_1~+~\left\|V_{\Omega^c}\right\|_1$ and rearranging, we get the required result.
\end{proof}
\noindent Now, we can use this to prove our main theorem.
\begin{proof}[Proof of Theorem \ref{thm.main}]
Let $\Omega$ be the support of $X$ and notice that $\Omega$ is $d-$distributed. Now, suppose that there exists an $\tilde{X}\neq X$ such that $A\tilde{X}A^T = Y$. 
Observe that $A\left(\tilde{X}-X\right)A^T = 0$. Now, consider 
\begin{align}
\left\|X\right\|_1&\leq\left\|X-\tilde{X}_\Omega\right\|_1+\left\|\tilde{X}_\Omega\right\|_1\notag\\
&= \left\|\left(X-\tilde{X}\right)_\Omega\right\|_1+\left\|\tilde{X}_\Omega\right\|_1\notag\\
&\stackrel{(a)}{\leq}{\frac{\epsilon}{1-3\epsilon}}\left\|\left(X-\tilde{X}\right)_{\Omega^c}\right\|+\left\|\tilde{X}_\Omega\right\|_1\notag\\
& = {\frac{\epsilon}{1-3\epsilon}}\left\|-\tilde{X}_{\Omega^c}\right\|+\left\|\tilde{X}_\Omega\right\|_1\notag\\
&<\left\|\tilde{X}\right\|_1\notag
\end{align}
where ${(a)}$ follows from Proposition~\ref{prop.nsp}, and the last line follows from the fact that $\epsilon < \frac{1}{4}$, again from Proposition~\ref{prop.nsp}. Therefore, the unique solution of \eqref{eq.optP1} is $X$ with probability exceeding $1-p^{-c}$, for some $c>0$. 
\end{proof}


\section{Conclusions}

In this paper we have introduced the notion of distributed sparsity for matrices. We have shown that when a matrix is $X$ distributed sparse, and $A,B$ are suitable random binary matrices, then it is possible to recover $X$ from under-determined linear measurements of the form $Y=AX B^T$ via $\ell_1$ minimization. We have also shown that this recovery procedure is robust in the sense that if $X$ is equal to a distributed sparse matrix plus a perturbation, then our procedure returns an approximation with accuracy proportional to the size of the perturbation. Our results follow from a new lemma about the properties of tensor products of random bipartite graphs. We also describe three interesting applications where our results would be directly applicable.

In future work, we plan to investigate the statistical behavior and sample complexity of estimating a distributed sparse matrix (and its exact support) in the presence of various sources of noise (such as additive Gaussian noise, and Wishart noise). We expect an interesting trade-off between the sketching dimension and the sample complexity. 


\bibliographystyle{plain}
\bibliography{refs}

\begin{thebibliography}{10}

\bibitem{graph_sketching}
Kook~Jin Ahn, Sudipto Guha, and Andrew McGregor.
\newblock Graph sketches: sparsification, spanners, and subgraphs.
\newblock In {\em Proceedings of the 31st symposium on Principles of Database
  Systems}, pages 5--14. ACM, 2012.

\bibitem{Andoni09efficientsketches}
Alexandr Andoni, Khanh~Do Ba, Piotr Indyk, and David Woodruff.
\newblock Efficient sketches for earth-mover distance, with applications.
\newblock In {\em in FOCS}, 2009.

\bibitem{angluin1979fast}
Dana Angluin and Leslie~G. Valiant.
\newblock Fast probabilistic algorithms for hamiltonian circuits and matchings.
\newblock {\em Journal of Computer and system Sciences}, 18(2):155--193, 1979.

\bibitem{Recht}
L.~Balzano, R.~Nowak, and B.~Recht.
\newblock Online identification and tracking of subspaces from highly
  incomplete information.
\newblock In {\em Proceedings of the 48th Annual Allerton Conference}, 2010.

\bibitem{berindeIndyk2008}
R.~Berinde, A.C. Gilbert, P.~Indyk, H.~Karloff, and M.J. Strauss.
\newblock Combining geometry and combinatorics: A unified approach to sparse
  signal recovery.
\newblock In {\em Communication, Control, and Computing, 2008 46th Annual
  Allerton Conference on}, pages 798 --805, sept. 2008.

\bibitem{bertsimas1997introduction}
Dimitris Bertsimas and John~N Tsitsiklis.
\newblock {\em Introduction to linear optimization}.
\newblock Athena Scientific Belmont, MA, 1997.

\bibitem{BickelLevina2}
P.~J. Bickel and E.~Levina.
\newblock Covariance regularization by thresholding.
\newblock {\em Annals of Statistics}, 36(6):2577--2604, 2008.

\bibitem{BickelLevina1}
P.~J. Bickel and E.~Levina.
\newblock Regularized estimation of large covariance matrices.
\newblock {\em Annals of Statistics}, 36(1):199--227, 2008.

\bibitem{bollobas2001random}
B.~Bollob{\'a}s.
\newblock {\em Random graphs}, volume~73.
\newblock Cambridge university press, 2001.

\bibitem{candes2006robust}
E.J. Cand{\`e}s, J.~Romberg, and T.~Tao.
\newblock Robust uncertainty principles: Exact signal reconstruction from
  highly incomplete frequency information.
\newblock {\em Information Theory, IEEE Transactions on}, 52(2):489--509, 2006.

\bibitem{candes2005decoding}
E.J. Candes and T.~Tao.
\newblock Decoding by linear programming.
\newblock {\em Information Theory, IEEE Transactions on}, 51(12):4203--4215,
  2005.

\bibitem{RobustPCA}
Emmanuel~J Candes, Xiaodong Li, Yi~Ma, and John Wright.
\newblock Robust principal component analysis?
\newblock {\em arXiv preprint arXiv:0912.3599}, 2009.

\bibitem{venkatCnvxGeometryArxiv2010}
V.~{Chandrasekaran}, B.~{Recht}, P.~A. {Parrilo}, and A.~S. {Willsky}.
\newblock {The Convex Geometry of Linear Inverse Problems}.
\newblock {\em ArXiv e-prints}, December 2010.

\bibitem{venkatConvexFOCM2012}
Venkat Chandrasekaran, Benjamin Recht, PabloA. Parrilo, and AlanS. Willsky.
\newblock The convex geometry of linear inverse problems.
\newblock {\em Foundations of Computational Mathematics}, 12:805--849, 2012.

\bibitem{Cohen2009}
Albert Cohen, Wolfgang Dahmen, and Ronald Devore.
\newblock {COMPRESSED SENSING AND BEST k-TERM APPROXIMATION}.
\newblock {\em Journal of the American Mathematical Society}, 22(1):211--231,
  2009.

\bibitem{cormode2006combinatorial}
Graham Cormode and S~Muthukrishnan.
\newblock Combinatorial algorithms for compressed sensing.
\newblock {\em Structural Information and Communication Complexity}, pages
  280--294, 2006.

\bibitem{diggle1998nonparametric}
Peter~J Diggle and Ar{\=u}nas~P Verbyla.
\newblock Nonparametric estimation of covariance structure in longitudinal
  data.
\newblock {\em Biometrics}, pages 401--415, 1998.

\bibitem{donoho2003optimally}
David~L Donoho and Michael Elad.
\newblock Optimally sparse representation in general (nonorthogonal)
  dictionaries via ?1 minimization.
\newblock {\em Proceedings of the National Academy of Sciences},
  100(5):2197--2202, 2003.

\bibitem{Donoho2001}
David~L Donoho and Xiaoming Huo.
\newblock {Uncertaintly Principles and Ideal Atomic Decomposition}.
\newblock {\em IEEE Transactions on Information Theory}, 47(7):2845--2862,
  2001.

\bibitem{donoho2006compressed}
David~Leigh Donoho.
\newblock Compressed sensing.
\newblock {\em Information Theory, IEEE Transactions on}, 52(4):1289--1306,
  2006.

\bibitem{duarte2012kronecker}
Marco~F Duarte and Richard~G Baraniuk.
\newblock Kronecker compressive sensing.
\newblock {\em Image Processing, IEEE Transactions on}, 21(2):494--504, 2012.

\bibitem{duarteKroneckerCS2010}
M.F. Duarte and R.G. Baraniuk.
\newblock Kronecker product matrices for compressive sensing.
\newblock In {\em Acoustics Speech and Signal Processing (ICASSP), 2010 IEEE
  International Conference on}, pages 3650 --3653, march 2010.

\bibitem{Fan_Fan_Lv_2007}
J.~Fan, Y.~Fan, and J.~Lv.
\newblock High dimensional covariance matrix estimation using a factor model.
\newblock {\em Journal of Econometrics}, 147(1):43, 2007.

\bibitem{network_compression}
Anna~C Gilbert and Kirill Levchenko.
\newblock Compressing network graphs.
\newblock In {\em Proceedings of the LinkKDD workshop at the 10th ACM
  Conference on KDD}. Citeseer, 2004.

\bibitem{gb08}
M.~Grant and S.~Boyd.
\newblock Graph implementations for nonsmooth convex programs.
\newblock In V.~Blondel, S.~Boyd, and H.~Kimura, editors, {\em Recent Advances
  in Learning and Control}, Lecture Notes in Control and Information Sciences,
  pages 95--110. Springer-Verlag Limited, 2008.

\bibitem{gribonval2003sparse}
R{\'e}mi Gribonval and Morten Nielsen.
\newblock Sparse representations in unions of bases.
\newblock {\em Information Theory, IEEE Transactions on}, 49(12):3320--3325,
  2003.

\bibitem{hajnal1970proof}
Andr{\'a}s Hajnal and Endre Szemer{\'e}di.
\newblock Proof of a conjecture of erdos.
\newblock {\em Combinatorial theory and its applications}, 2:601--623, 1970.

\bibitem{haupt2010toeplitz}
Jarvis Haupt, Waheed~U Bajwa, Gil Raz, and Robert Nowak.
\newblock Toeplitz compressed sensing matrices with applications to sparse
  channel estimation.
\newblock {\em Information Theory, IEEE Transactions on}, 56(11):5862--5875,
  2010.

\bibitem{hedges1985statistical}
Larry~V Hedges, Ingram Olkin, Mathematischer Statistiker, Ingram Olkin, and
  Ingram Olkin.
\newblock {\em Statistical methods for meta-analysis}.
\newblock Academic Press New York, 1985.

\bibitem{cvx}
CVX~Research{,} Inc.
\newblock {CVX}: Matlab software for disciplined convex programming, version
  2.0 beta, September 2012.

\bibitem{johnson1984extensions}
W.B. Johnson and J.~Lindenstrauss.
\newblock Extensions of {L}ipschitz mappings into a {H}ilbert space.
\newblock {\em Contemporary mathematics}, 26(189-206):1--1, 1984.

\bibitem{jokarKroneckerCS2010}
S.~Jokar.
\newblock Sparse recovery and kronecker products.
\newblock In {\em Information Sciences and Systems (CISS), 2010 44th Annual
  Conference on}, pages 1 --4, march 2010.

\bibitem{jokar2009sparse}
Sadegh Jokar and Volker Mehrmann.
\newblock Sparse solutions to underdetermined kronecker product systems.
\newblock {\em Linear Algebra and its Applications}, 431(12):2437--2447, 2009.

\bibitem{Kane:2011:FME:1993636.1993735}
Daniel~M. Kane, Jelani Nelson, Ely Porat, and David~P. Woodruff.
\newblock Fast moment estimation in data streams in optimal space.
\newblock In {\em Proceedings of the 43rd annual ACM symposium on Theory of
  computing}, STOC '11, pages 745--754, New York, NY, USA, 2011. ACM.

\bibitem{journals/corr/abs-0902-4045}
M.~Amin Khajehnejad, Alexandros~G. Dimakis, Weiyu Xu, and Babak Hassibi.
\newblock Sparse recovery of positive signals with minimal expansion.
\newblock {\em CoRR}, abs/0902.4045, 2009.

\bibitem{khajehnejad2011sparse}
MAmin Khajehnejad, Alexandros~G Dimakis, Weiyu Xu, and Babak Hassibi.
\newblock Sparse recovery of nonnegative signals with minimal expansion.
\newblock {\em Signal Processing, IEEE Transactions on}, 59(1):196--208, 2011.

\bibitem{Kierstead:2008:SPH:1348972.1348978}
H.~A. Kierstead and A.~V. Kostochka.
\newblock A short proof of the {H}ajnal-{S}zemeredi {T}heorem on equitable
  colouring.
\newblock {\em Comb. Probab. Comput.}, 17(2):265--270, March 2008.

\bibitem{Lauritzen_book}
S.~Lauritzen.
\newblock {\em Graphical Models}.
\newblock Clarendon Press, Oxford, 1996.

\bibitem{meinshausenBuhlmann2006}
Nicolai Meinshausen and Peter B{\"u}hlmann.
\newblock High-dimensional graphs and variable selection with the lasso.
\newblock {\em The Annals of Statistics}, 34(3):1436--1462, 2006.

\bibitem{muthukrishnan2005data}
S~Muthukrishnan.
\newblock {\em Data streams: Algorithms and applications}.
\newblock Now Publishers Inc, 2005.

\bibitem{Pemmaraju:2001:ECE:365411.365811}
Sriram~V. Pemmaraju.
\newblock Equitable colorings extend {C}hernoff-{H}oeffding bounds.
\newblock In {\em Proceedings of the twelfth annual ACM-SIAM symposium on
  Discrete algorithms}, SODA '01, pages 924--925, Philadelphia, PA, USA, 2001.
  Society for Industrial and Applied Mathematics.

\bibitem{RWRY}
P.~Ravikumar, M.~J. Wainwright, G.~Raskutti, and B.~Yu.
\newblock High-dimensional covariance estimation by minimizing
  $\ell_1$-penalized log-determinant divergence.
\newblock {\em Electronic Journal of Statistics}, 5, 2011.

\bibitem{rothman2009generalized}
Adam~J Rothman, Elizaveta Levina, and Ji~Zhu.
\newblock Generalized thresholding of large covariance matrices.
\newblock {\em Journal of the American Statistical Association},
  104(485):177--186, 2009.

\bibitem{rudelson2008sparse}
Mark Rudelson and Roman Vershynin.
\newblock On sparse reconstruction from fourier and gaussian measurements.
\newblock {\em Communications on Pure and Applied Mathematics},
  61(8):1025--1045, 2008.

\bibitem{sachs2005causal}
Karen Sachs, Omar Perez, Dana Pe'er, Douglas~A Lauffenburger, and Garry~P
  Nolan.
\newblock Causal protein-signaling networks derived from multiparameter
  single-cell data.
\newblock {\em Science Signalling}, 308(5721):523, 2005.

\bibitem{tropp2004greed}
Joel~A Tropp.
\newblock Greed is good: Algorithmic results for sparse approximation.
\newblock {\em Information Theory, IEEE Transactions on}, 50(10):2231--2242,
  2004.

\bibitem{vershynin2010introduction}
Roman Vershynin.
\newblock Introduction to the non-asymptotic analysis of random matrices.
\newblock {\em arXiv preprint arXiv:1011.3027}, 2010.

\bibitem{wainwright2009sharp}
Martin~J Wainwright.
\newblock Sharp thresholds for high-dimensional and noisy sparsity recovery
  using $\ell_1-$constrained quadratic programming (lasso).
\newblock {\em Information Theory, IEEE Transactions on}, 55(5):2183--2202,
  2009.

\end{thebibliography}


\renewcommand\appendixname{Appendix}
\begin{appendices}
\section{}
\subsection*{Proof of Lemma~\ref{lemma.main}}
\label{appendix:A}

\begin{lema}
Suppose that $G_1 = ([p],[m],E_1)$ and $G_2 = ([p],[m],E_2)$ are two independent uniformly random $\delta-$left regular bipartite graphs with  $\delta = \mathcal{O}(\log p)$ and $m=\mathcal{O}(\sqrt{dp} \log p)$. Let $\Omega\in\mathfrak{W}_{d,p}$ be fixed. Then there exists an $\epsilon \in \left( 0,\frac{1}{4} \right)$ such that $G_1\otimes G_2$ has the following properties with probability exceeding $1-p^{-c}$, for some $c>0$. 
\begin{enumerate}
\item  $|N(\Omega)| \geq p \delta^2(1-\epsilon)$.
\item For any $(i,i') \in ([p] \times [p]) \setminus \Omega $ we have $|N(i,i')\cap N(\Omega)| \leq \epsilon \delta^2$.  
\item For any $(i,i') \in \Omega$, $|N(i,i')\cap N(\Omega \setminus (i,i'))| \leq \epsilon \delta^2$.  
\end{enumerate}
Moreover, all these claims continue to hold when $G_2$ is the same as $G_1$. 
\end{lema}
\begin{proof}
As stated earlier, we will only prove this lemma for the case when $G_1=G_2$. With a few minor modifications, one can readily get a proof for the easier case when $G_1$ and $G_2$ are drawn independently. 

Let $\mathcal{E}_1, \mathcal{E}_2$ and, $\mathcal{E}_3$ respectively denote the events that the implications (1), (2) and, (3) are true. Notice that 
\begin{align*}
\mathbb{P}\left(\mathcal{E}_1^c\cup \mathcal{E}_2^c\cup \mathcal{E}_3^c\right)&\leq \mathbb{P}(\mathcal{E}_1^c)+\mathbb{P}(\mathcal{E}_2^c)+\mathbb{P}(\mathcal{E}_3^c) \\
& = \mathbb{P}(\mathcal{E}_1^c) + \mathbb{P}(\mathcal{E}_2^c\mid \mathcal{E}_1)\mathbb{P}(\mathcal{E}_1) + \mathbb{P}(\mathcal{E}_2^c\mid \mathcal{E}_1^c)\mathbb{P}(\mathcal{E}_1^c)\\&\qquad\qquad\hspace{-1mm} + \mathbb{P}(\mathcal{E}_3^c\mid \mathcal{E}_1)\mathbb{P}(\mathcal{E}_1) + \mathbb{P}(\mathcal{E}_3^c\mid \mathcal{E}_1^c)\mathbb{P}(\mathcal{E}_1^c)
\\
&\leq 3\mathbb{P}(\mathcal{E}_1^c) + \mathbb{P}(\mathcal{E}_2^c\mid \mathcal{E}_1) + \mathbb{P}(\mathcal{E}_3^c\mid \mathcal{E}_1)
\end{align*}

Our strategy will be to upper bound $\mathbb{P}(\mathcal{E}_1^c), \mathbb{P}(\mathcal{E}_2^c\mid \mathcal{E}_1)$ and, $\mathbb{P}(\mathcal{E}_3^c\mid \mathcal{E}_1)$. Suppose the bounds were $p_1,p_2$ and, $p_3$ respectively, then it is easy to see that 
\begin{equation}
\mathbb{P}\left\{\mathcal{E}_1\cap \mathcal{E}_2\cap \mathcal{E}_3\right\}\geq 1-\max\{3p_1,p_2,p_3\}.\label{eq.Lemma1}
\end{equation}

\noindent\underline{Part 1}. We will first show that $\mathbb{P}(\mathcal{E}_1^c)$ is small. Since $\Omega$ is $d-$distributed, the ``diagonal'' set $\mathcal{D} := \left\{(1,1),\ldots, (p,p)\right\}$ is a subset of $\Omega$. Now, notice that for $j\neq j'\in [m], i\in [p]$, 
\begin{equation}
\mathbb{P}\left[(j,j')\in N((i,i))\right]
= \frac{\delta(\delta-1)}{m(m-1)}
\end{equation}
This implies that 
$$\mathbb{P}\left[(j,j')\notin N(\mathcal{D})\right]
 = \left(1-\frac{\delta(\delta-1)}{m(m-1)}\right)^{\left|\mathcal{D}\right|}
$$
Therefore, we can bound the expected value of $\left|N(\Omega)\right|$ as follows. 
\begin{align*}
\mathbb{E}\big[\left|N(\Omega)\right|\big]
&\geq \mathbb{E}\big[\left|N\left(\mathcal{D}\right)\right|\big]\\
& = \sum_{jj'\in[m]\times [m]}\mathbb{P}\left[(j,j')\in N(\mathcal{D})\right]\\
&\geq\sum_{\substack{jj'\in[m]\times[m],\\ j\neq j'}}\mathbb{P}\left[(j,j')\in N(\mathcal{D})\right]\\
&= \sum_{\substack{jj'\in[m]\times[m],\\ j\neq j'}} \left(1 - \left(1-\frac{\delta(\delta-1)}{m(m-1)}\right)^{\left|\mathcal{D}\right|}\right)\\
&= m(m-1) \left(1 - \left(1-\frac{\delta(\delta-1)}{m(m-1)}\right)^{\left|\mathcal{D}\right|}\right)\\
&\geq m(m-1) \left(\frac{\left|\mathcal{D}\right|\delta(\delta-1)}{m(m-1)}-\frac{\left|\mathcal{D}\right|^2\delta^2(\delta-1)^2}{m^2(m-1)^2}\right)\\
& = \left|\mathcal{D}\right|\delta^2\left(1 - \left(\frac{1}{\delta}+\frac{(\delta-1)^2\left|\mathcal{D}\right|}{m(m-1)}\right)\right)\\
& = p\delta^2\left(1-\epsilon'\right).
\end{align*}
Where in the last step, we set $\epsilon' = \frac{1}{\delta}+\frac{(\delta-1)^2\left|\mathcal{D}\right|}{m(m-1)}$.

To complete the proof, we must show that the random quantity $\left|N(\Omega)\right|$ cannot be much smaller than $p\delta^2(1-\epsilon')$. As a first step, we define the random variables $\chi_{jj'} :=\mathbf{1}_{\left\{(j,j')\in N(\mathcal{D})\right\}}$ and notice that the following chain of inequalities hold
$$\left|N(\Omega)\right|\geq\left|N(\mathcal{D})\right|\geq\sum_{\substack{jj'\in[m]\times [m]\\j\neq j'}}\chi_{jj'}.$$
Therefore, we have that 
$$\mathbb{P}\left[\left|N(\Omega)\right|<p\delta^2(1-\epsilon'-\epsilon'')\right]\leq\mathbb{P}\left[\sum_{\substack{jj'\in[m]\times [m]\\j\neq j'}}\chi_{jj'}<p\delta^2(1-\epsilon'-\epsilon'')\right].$$
Also, since by above, $\mathbb{E}\left[\sum_{j\neq j'}\chi_{jj'}\right]\geq p\delta^2(1-\epsilon')$, we have that 
$$\mathbb{P}\left[\left|N(\Omega)\right|<p\delta^2(1-\epsilon)\right]\leq\mathbb{P}\left[\sum_{\substack{jj'\in[m]\times [m]\\j\neq j'}}\chi_{jj'}<\mathbb{E}\left[\sum_{\substack{jj'\in[m]\times [m]\\j\neq j'}}\chi_{jj'}\right]-p\delta^2\epsilon''\right].$$

Now, notice that the sum $\sum_{j\neq j'}\chi_{jj'}$ has $m(m-1)$ terms and each term in the sum is dependent on no more than $2m-4$ terms. Therefore, one way to bound the required quantity is to extract independent sub-sums from the above sum and bound the deviation of each of those from their means (which is the corresponding sub-sum of the mean). A principled way of doing this is suggested by the celebrated Hajnal-Szemeridi theorem \cite{hajnal1970proof,Kierstead:2008:SPH:1348972.1348978}. Consider a graph on the vertex set $[m]\times [m]\setminus\{(1,1),\ldots,(m,m)\}$ where there is an edge between vertices $(j,j')$ and $(j_1,j_1')$ if $j=j_1$ and/or $j'=j_1'$, i.e., exactly when the random variables $\chi_{jj'}$ and $\chi_{j_1j_1'}$ are dependent. Since this graph has degree $\Theta(m)$, Hajnal-Szemeridi theorem tells us that this graph can be equitable colored with $\Theta(m)$ colors. In other words, the above sum can be partitioned into $\Theta(m)$ sub-sums such that each sub-sum has $\Theta(m)$ elements and the random variables in each of them are independent. Along with this and the fact that $m(m-1)>p\delta^2$, we can use the union bound and write
\begin{align*} \small
\mathbb{P}\left[\sum_{\substack{jj'\in[m]\times [m]\\j\neq j'}}\chi_{jj'}<\mathbb{E}\left[\sum_{\substack{jj'\in[m]\times [m]\\j\neq j'}}\chi_{jj'}\right]-p\delta^2\epsilon''\right]& \\
& \hspace {-48mm}\leq \Theta(m)\; \mathbb{P}\left[\frac{1}{\left|C_1\right|}\sum_{jj'\in C_1}\chi_{jj'}<\mathbb{E}\left[\frac{1}{\left|C_1\right|}\sum_{jj'\in C_1}\chi_{jj'}\right]-\epsilon''\right],
\end{align*} \normalsize
where $C_1$ is one of the ``colors''. Notice that $\left|C_1\right|=\Theta(m)$. 

Finally, using Chernoff bounds, we have
\begin{align*}
&\Theta(m)\;\mathbb{P}\left[\frac{1}{\left|C_1\right|}\sum_{jj'\in C_1}\chi_{jj'}<\frac{1}{\left|C_1\right|}\mathbb{E}\left[\sum_{jj'\in C_1}\chi_{jj'}\right]-\epsilon''\right]\\&\hspace{85.5mm}\leq \Theta(m)\; \exp\left\{-2\epsilon''^2\Theta(m)\right\}.
\end{align*}
 Finally, since $\epsilon'$ can be made as small as possible, setting $\epsilon :=\epsilon'+\epsilon''$ yeilds $p_1 < p^{-c_1}$ for some $c_1 > 0$. This  technique of generating large deviation bounds when one has limited dependence is not new, see \cite{Pemmaraju:2001:ECE:365411.365811}.

\noindent \underline{Part 2:} Now, we bound $\mathbb{P}(\mathcal{E}_2^c\mid \mathcal{E}_1)$. Associated to a fixed  $i$ one can imagine $\delta$ independent random trials that determine the outgoing edges from $i$. In a similar way there are $\delta$ independent random trials associated to the outgoing edges of $i'$. Let us fix $(i,i')$ and investigate the outgoing edge (in the tensor graph) determined by the first trial of $i$ and the first trial of $i'$. The probability that this edge emanating from the vertex $(i,i')\in[p]^2\setminus\{(1,1),\ldots,(p,p)\}$ hits an arbitrary vertex $(j,j')\in [m]^2$ is given by $1/m^2$. The probability that this edge lands in $N(\Omega)$ is, therefore, given by $\left|N(\Omega)\right|/m^2$. Since there are $\delta^2$ edges that are incident on $(i,i')$, the expected size of overlap between $N(i,i')$ and $N(\Omega)$ is upper bounded by $$\delta^2\frac{\left|N(\Omega)\right|}{m^2}.$$

Again, to show concentration, we employ similar arguments as before and define indicator random variables $\chi_1\ldots, \chi_{\delta^2}$ each of which corresponds to one of the edges emanating from the vertex $(i,i')$ and then observing that the sum $\sum_{k=1}^{\delta^2}\chi_k$ is precisely equal to the random quantity $\left|N(i,i')\cap N(\Omega)\right|$. To conclude that this random quantity concentrates, we first observe that, as above, the $\delta^2$ dependent terms can be divided up into $\Theta(\delta)$ with $\Theta(\delta)$ elements each such that in each set the terms are independent. Therefore, we have 
\begin{align*}
\mathbb{P}\left[\left|N(i,i')\cap N(\Omega)\right|>\delta^2\frac{\left|N(\Omega)\right|}{m^2}(1+\epsilon')\right]
&= \mathbb{P}\left[\sum_{k=1}^{\delta^2}\chi_k>\delta^2\frac{\left|N(\Omega)\right|}{m^2}(1+\epsilon')\right]\\
&\leq \Theta(\delta)\;\mathbb{P}\left[\sum_{k\in C_1}\chi_k > \Theta(\delta)\frac{\left|N(\Omega)\right|}{m^2}(1+\epsilon')\right]\\
&\leq \Theta(\delta)\exp\left\{-\delta\frac{\left|N(\Omega)\right|}{m^2}\epsilon'\right\},
\end{align*}
where $C_1$ is one of the colors. 

Therefore, since conditioned on $\mathcal{E}_1$, $|N(\Omega)|>\delta^2 p (1-\epsilon)$ if we pick $m= \delta \sqrt{dp}$, and $\delta=\Theta(\log p)$ there is a $c_2'=c_2'(\epsilon')>2$ such that,  $\left|N(i,i')\cap N(\Omega)\right|>\delta^2\frac{\left|N(\Omega)\right|}{m^2}(1+\epsilon')$, with probability not exceeding $p^{-c_2'}$. Setting $\epsilon = (1+\epsilon')\left|N(\Omega)\right|/m^2$, picking $m$ as prescribed, and taking union bound over $(i,i') \in \Omega^c$, we get $p_2 \leq p^{-c_2}$ for some $c_2>0$. 

\noindent\underline{Part 3:} 
Next, we bound $\mathbb{P}(\mathcal{E}_3^c\mid \mathcal{E}_1)$. Notice that the proof is very similar to that of part $2$ when $i\neq i'$. So, here we will consider the quantity $\left|N(i,i)\cap N(\Omega\setminus\{(i,i)\})\right|$.

As explained above to each left node $i$ we associate $\delta$ random trials that determine its outgoing edges. Correspondingly, if we fix a left node $(i,i)$ in the tensor graph, and think of its outgoing edges they are determined by  the outcome of $\delta^2$ product trials. Let $\ind_k(j)$ be the indicator function of the event that in the $k^{th}$ trial of $i$ the outgoing edge is incident on $j$. The probability that the edge associated to the $(k,l)$ trial associated to $(i,i)$ is incident on $(j,j')$ is the random variable $\ind_k(j) \ind_l(j')$. Note that $\ind_k(j) \ind_k(j')= \ind_k(j)$ if $j=j'$ and $0$ otherwise.

Note that 
\begin{align*}
|N(i,i) \cap N(\Omega \setminus (i,i))| &= \sum_{k=1}^{\delta} \sum_{l=1}^{\delta} \sum_{(j,j')\in N(\Omega \setminus (i,i))} \ind_k(j) \ind_l(j') \\
& \leq \delta + \sum_{k \neq l} \sum_{(j,j')\in N(\Omega \setminus (i,i))} \ind_k(j) \ind_l(j')
\end{align*}
When $k \neq l$, the trials corresponding to  $\ind_k(j), \ind_l(j')$ are independent, and hence $\mathbb{E}\ind_k(j) \ind_l(j') = \frac{1}{m^2}.$ We define $\chi_{k,l}:=\sum_{(j,j')\in N(\Omega \setminus (i,i))} \ind_k(j) \ind_l(j')$ and note that $\mathbb{E} \left( \chi_{k,l} \right) = \frac{N(\Omega \setminus (i,i))}{m^2}$. We also note that $\chi_{k,l}$ is binary valued, and
$$|N(i,i) \cap N(\Omega \setminus (i,i))| \leq \delta+ \sum_{k \neq l} \chi_{k,l}.$$

Therefore, 
\begin{align*}\mathbb{E} \left( |N(i,i) \cap N(\Omega \setminus (i,i))|\right) &\leq \delta+ (\delta^2-\delta)\frac{N(\Omega \setminus (i,i))}{m^2} \\
& \leq \delta + (\delta^2-\delta)\frac{\delta^2 dp}{m^2} \\
& \leq \delta^2 \left(\frac{1}{\delta} + \left(1 - \frac{1}{\delta}\right) \frac{\delta^2 dp}{m^2} \right) \\
& \leq  \delta^2 \epsilon.
\end{align*}

Next we need to prove that the quantity of interest $\sum_{k \neq l} \chi_{k,l}$ concentrates about its mean. To that end we note that these binary valued variables are such that any particular $\chi_{k,l}$ is dependent on at most $2 \delta -2$ other variables. Using Chernoff concentration bounds in conjunction with the Hajnal-Szemeredi based coloring argument explained in part 1 of this proof, followed by a union bound over $i \in [p]$ we obtain the required probability bounds $p_3 < p^{-c_3}$ for some $c_3>0$.

Substituting the bounds for $p_1, p_2, p_3$ back into \eqref{eq.Lemma1} concludes the proof.
\end{proof}


\section{}
\subsection*{Proof of Proposition~\ref{prop.random_sparsity}} 
\label{appendix.B}
\begin{prop1} 
Consider a random matrix $X \in\mathbb{R}^{p \times p}$ such that $X_{ij}\stackrel{iid}{\sim} {\rm Ber}(\gamma)$ where $p\gamma = \Delta= \Theta(1)$, then for any $\epsilon>0$, $X$ is $d-$distributed sparse with probability at least $1-\epsilon$, where $$d = \Delta\left(1 + {\frac{2\log(2p/\epsilon)}{\Delta}}\right).$$
\end{prop1}

\begin{proof}
Let $X_i, i=1,\ldots,p$ denote the sparsity of the $i-$th column and let $X_i, i = p+1,\ldots, 2p$, denote the sparsity of the $i-$th row. Notice that the $2p$ random variables $X_1,X_2,\ldots,X_{2p}$ are (dependent) Bin$(p,\gamma)$ random variables. With the choice of $d$ as indicated in the theorem, we have the following, 
\begin{align*}
\mathbb{P}(X_1 > d) & = \mathbb{P}\left(X_1 > \Delta\left(1 +{\frac{\log (2p/\epsilon)}{\Delta}} \right)\right)\\
&\stackrel{(a)}{\leq} \exp\left\{-\frac{\beta^2 \Delta}{2+\beta}\right\},\;\;\;\; \beta = \frac{2\log(2p/\epsilon)}{\Delta}\\
&\stackrel{(b)}{\leq} \exp\left\{-\frac{\beta \Delta}{2}\right\}\\
& = \frac{\epsilon}{2p}
\end{align*}
where $(a)$ follows from the multiplicative form of the Chernoff Bound \cite{angluin1979fast} and $(b)$ follows as long as $\beta > 2$. The rest of the proof follows from a simple application of the union bound.  
\end{proof}


\section{}
\label{appendix:C}
\subsection*{Proof of Theorem~\ref{thm.approximation}}
\setcounter{theorem}{1}

\begin{theorem}
\label{thm.approximation}
Suppose that $X$ is a $p\times p$ matrix. Furthermore, suppose that the hypotheses of Theorem~\ref{thm.main} hold and let $X^\ast$ be the solution to the optimization program~\eqref{eq.optP1}. Then, there exists a $c>0$ and an $\epsilon\in(0,1/4)$ such that the following holds with probability exceeding $1-p^{-c}$. 
\begin{equation}
\left\|X^\ast - X\right\|_1 \leq \frac{2 - 4\epsilon}{1-4\epsilon}\left(\min_{\Omega\in\mathfrak{W}_{d,p}} \left\|X - X_\Omega\right\|_1\right).
\end{equation}
\end{theorem}
\begin{proof}
Since $X^\ast$ is the optimum of the optimization program~\eqref{eq.optP1}, we have that $\left\|X\right\|_1\geq \left\|X^\ast\right\|_1$. Let $\Omega^\ast$ be such that $\|X-X_{\Omega^\ast}\|_1 = \min_{\Omega\in\mathfrak{W}_{d,p}}\|X-X_X\|$. We can proceed as follows
\begin{align}
\left\|X\right\|_1 & \geq \left\|X^\ast\right\|_1\\
& = \left\|(X + X^\ast - X)_\Omega\right\|_1 + \left\|(X + X^\ast - X)_{\Omega^c}\right\|_1\\
& \geq \left\|X_\Omega\right\|_1 - \left\|(X^\ast - X)_\Omega\right\|_1 + \left\|(X^\ast - X)_{\Omega^c}\right\|_1 - \left\|X_{\Omega^c}\right\|_1\\
& = \left\|X\right\|_1 - 2\left\|X_{\Omega^c}\right\|_1 + \left\|X^\ast - X\right\|_1 - 2\left\|(X-X^\ast)_\Omega\right\|_1\\
& \geq \left\|X\right\|_1 - 2\left\|X_{\Omega^c}\right\|_1 + \left(1 - \frac{2\epsilon}{1-2\epsilon}\right)\left\|X^\ast - X\right\|_1
\end{align}
where in the last step, we have used the fact that since $X^{\ast}$ is a feasible point in \eqref{eq.optP1}, $AX^{\ast} B^T = AX B^T$ and therefore, we can apply the result of Proposition~\ref{prop.nsp} to $X^{\ast} - X$. This completes the proof. 
\end{proof}
\end{appendices}

\end{document}